\documentclass[a4paper,UKenglish,cleveref, autoref, thm-restate]{lipics-v2019}

\title{Relaxing Common Belief for Social Networks} 

\titlerunning{Relaxing Common Belief} 

\author{Noah Burrell}{University of Michigan, Ann Arbor, USA}{burrelln@umich.edu}{https://orcid.org/0000-0003-3448-080X}{}

\author{Grant Schoenebeck}{University of Michigan, Ann Arbor, USA}{schoeneb@umich.edu}{https://orcid.org/0000-0001-6878-0670}{}

\authorrunning{N. Burrell and G. Schoenebeck} 

\Copyright{Noah Burrell and Grant Schoenebeck} 



\category{} 

\relatedversion{
An abridged version of the paper was published in the proceedings of the 12th Innovations in Theoretical Computer Science (ITCS) conference.} 

\supplement{}

\funding{The authors gratefully acknowledge the support of the National Science Foundation under Career Award \#1452915.}

\acknowledgements{}

\nolinenumbers 

\hideLIPIcs  

\EventEditors{John Q. Open and Joan R. Access}
\EventNoEds{2}
\EventLongTitle{42nd Conference on Very Important Topics (CVIT 2016)}
\EventShortTitle{CVIT 2016}
\EventAcronym{CVIT}
\EventYear{2016}
\EventDate{December 24--27, 2016}
\EventLocation{Little Whinging, United Kingdom}
\EventLogo{}
\SeriesVolume{42}
\ArticleNo{23}

\usepackage{mathtools}
\usepackage[ruled]{algorithm2e}
\usepackage{bm}
\usepackage{graphicx}

\graphicspath{{./graphics/}}

\DeclarePairedDelimiter{\set}{\{}{\}}

\newcommand{\sse}{\subseteq}
\newcommand{\E}{\mathrm{E}}

\begin{document}

\maketitle

\begin{abstract}
We propose a relaxation of common belief called \textit{factional belief} that is suitable for the analysis of strategic coordination on social networks. We show how this definition can be used to analyze revolt games on general graphs, including by giving an efficient algorithm that characterizes a structural result about the possible equilibria of such games.

This extends prior work on common knowledge and common belief, which has been too restrictive for use in understanding strategic coordination and cooperation in social network settings.
\end{abstract}

\section{Introduction}
Common knowledge and its analogues (e.g. common belief) are fundamental in the analysis of coordination and cooperation in strategic settings. Informally, common knowledge is the phenomenon that, within a population, everyone knows that a proposition is true, everyone knows that everyone knows that it is true, everyone knows that everyone knows that everyone knows that it is true, and so on, ad infinitum. The existence of common knowledge often shows up as an assumption that underlies some kind of strategic coordination. For example, in the standard game theoretic setting, it is typically assumed that the agents playing a game have common knowledge both of the payoffs of the game and of the rationality of each agent. This assumption undergirds the agents' ability to coordinate on equilibria. 

A rigorous investigation into the consequences of this type of assumption can be traced back to Robert Aumann \cite{Aumann1976}, who showed that if two rational agents have the same prior, and their posteriors for an event are common knowledge, then their posteriors must be equal (i.e. it is impossible for such agents to ``agree to disagree''). Later work considered whether assumptions about common knowledge, and the mathematical consequences of those assumptions, were realistic for trying to explain economic phenomena. In particular, it seemed unrealistic to expect agents to reason about infinite hierarchies of knowledge. Initial attempts to address this critique involved truncating the infinite hierarchy that is required in our informal definition of common knowledge after some large, but finite, number of levels. However, this line of inquiry resulted in the discovery of ``common knowledge paradoxes'' that arose from examples like Ariel Rubinstein's Electronic Mail Game \cite{Rubinstein1989}. Such examples demonstrated that in situations where common knowledge was required for coordination, if the infinite hierarchy were truncated to be any finite hierarchy, strategic agents behaved unexpectedly and unrealistically. That is, they behaved very differently when a proposition was close to common knowledge, in the sense that many---but only finitely many---hierarchies of knowledge were satisfied, than they did when the proposition was actually common knowledge (and therefore infinitely many hierarchies of knowledge were satisfied).

An important breakthrough came from Dov Monderer and Dov Samet in 1989 \cite{Monderer1989}. Inspired by an earlier version of Rubinstein's work \cite{Rubinstein1989} and the work of Aumann \cite{Aumann1976}, Monderer and Samet proposed an alternative to truncating the infinite hierarchy of knowledge: relaxing the requirement of ``knowledge'' at each level in the hierarchy. They introduced the notion of  \textit{common belief}---an analogous concept to common knowledge that is defined by replacing ``everyone knows'' with ``everyone believes with probability (at least) $p$'' at each place where it occurs in our informal definition of common knowledge. Under this definition, they showed that common belief approximates common knowledge in precisely the way that truncating the infinite hierarchy did not: When common knowledge of some proposition is relaxed to common belief, agents behave in approximately the same way they did when they had common knowledge of the proposition. For example, common belief can be applied to generalize Aumann's result that was mentioned above: when rational agents have the same prior and their posteriors about an event are commonly believed, these posteriors must be approximately equal.

In the same paper, they also detailed a key insight that allowed them to address the critique that common knowledge might be too unrealistic to have explanatory power in Economics. They showed that common knowledge has an alternative definition (used implicitly in \cite{Aumann1976}), that is formally equivalent to the more natural definition, but simpler to reason about. This alternative---inspired by the example of public announcements---defines common knowledge in terms of events that are \textit{evident knowledge}: events for which their occurrence implies knowledge of their occurrence within the entire population. Common belief has a similar formally equivalent, alternative definition in terms of events that are \textit{evident belief}: events for which their occurrence implies belief with probability at least $p$ of their occurrence within the entire population.

Here, it is worth noting that our work most heavily borrows from that of Monderer and Samet and, as such, their paper \cite{Monderer1989} is highly recommended both as a primer for what follows here and as an introduction to common knowledge and common belief in general.

In a somewhat orthogonal line of research, Michael Suk-Young Chwe studied common knowledge (not common belief) on networks in the context of a revolt game \cite{Chwe1999, Chwe2000}. In his setting, each agent has a threshold that represents the size of the revolt in which she would agree to participate. For example, an agent with a threshold of 3 would need to know that at least 2 other agents would participate in order to herself agree to participate in the revolt. To decide whether this threshold is met, agents learn the thresholds of their neighbors. They use this information (and knowledge of their 2-hop neighborhood in the graph) to reason about which agents have their thresholds satisfied and consequently whether their own threshold is satisfied. Because agents require absolute certainty in this setting, they only consider their neighbors when conducting this reasoning. 

The strict requirement that absolute certainty is necessary for an agent to revolt implies that for a group of agents to revolt, it must be common knowledge among them that all of their thresholds are satisfied. Consequently, the problem of finding groups of revolting agents reduces to the problem of finding cliques of a certain size in the network. However, an important innovation of Chwe's work, in contrast to the settings described above, is that common knowledge in his setting is a local phenomenon occurring within those cliques, not a global phenomenon occurring within the entire population. 

Both approaches---that of Monderer and Samet and that of Chwe---are limited in their usefulness in social network settings, because their respective notions are unlikely to apply in many natural graphs that are used to model social networks. For example, an Erd\H{o}s-R{\'e}nyi random graph of $n$ vertices with $p = \frac{10}{n}$ will almost certainly be sparse, so population-level phenomenon like common belief will not arise. On the other hand, an Erd\H{o}s-R{\'e}nyi random graph of $n$ vertices with $p = \frac{1}{2}$ would be unlikely to contain large cliques, so the phenomenon of local common knowledge would be severely limited despite the fact that each agent, in seeing about half the graph, should have a lot of information about the entire population and might be expected to be able to coordinate with some large fraction of it. 

We propose that Monderer's and Samet's concept of common belief---itself a relaxation of common knowledge---can be further relaxed to a notion of common belief among a faction (i.e. a minimal-size subset) of the population, while retaining both its mathematical simplicity (in being defined in terms of events that are evident belief) and its economic explanatory power. We refer to this notion as \textit{factional belief}.  

Factional belief is a natural application of the ideas of Monderer and Samet to network settings similar to those of Chwe. It retains from Chwe the idea that common knowledge/belief can occur in only a subset of the population and still motivate their behavior. However, it is not prohibitively strict, such that it would be unlikely to occur endogenously in many natural graphs. Factional belief is not necessarily local---agents can and may need to reason about agents outside of their neighborhood. As such, it does not require cliques.

\subsection{Our Contributions}
\begin{itemize}
    \item We formally define a notion of factional belief (Section~\ref{section:defining}), which can be used to analyze revolt games on general graphs (Section~\ref{section:model}). Prior notions of common knowledge and common belief were insufficient for this type of analysis.
    
    \item We provide an algorithm that characterizes a structural result about the types of equilibria that are possible in instances of the network revolt games described in Section~\ref{section:model} (Section~\ref{section:applying}) and a few natural extensions of those games (Section~\ref{section:broadening}).
    
    \item We show that, surprisingly, it is sufficient for our algorithm to only have access to the degree sequence of the network; additional details of the network beyond the degree sequence are not relevant. 
    
    \item We demonstrate the practical utility of our algorithms from Sections \ref{section:applying} and \ref{section:broadening} by applying them to simulated network data to explore how various parameters of networks and of the model relate to the size of revolts that are supported in equilibria of the network revolt game (Section~\ref{section:experiments}).
\end{itemize}

\subsection{Additional Related Work}
The work of Stephen Morris is particularly notable when surveying the literature related to common knowledge and common belief. Morris, often following the work of Monderer and Samet, has done much theoretical work relating to common knowledge and common belief \cite{Morris1999, Morris2002, Morris2014, Morris1997}. More recently, he has also collaborated with Benjamin Golub to study higher-order reasoning, reminiscent of the infinite hierarchy of reasoning in the initial definitions of common knowledge and common belief, in network settings \cite{Golub2017b, Golub2017a}.

\vspace{1 ex}

Underlying Morris' work, above, and our work is our assertion that common knowledge, common belief, and factional belief are useful, not just as mathematical concepts, but for understanding real social and economic phenomena. Here, we briefly outline some relevant work in applying common knowledge and common belief to explaining such phenomena. One clear direction for future research is to try to similarly apply factional belief as an explanatory tool in these and other settings.

Morris has applied common knowledge/belief to settings such as contagion \cite{Morris2000} and global games \cite{Morris2016}.

Chwe, in 2013, revisited his earlier work and expanded his purview to consider how common knowledge is generated in society \cite{Chwe2013}. He proposed that the importance of rituals in society can be understood from the perspective that rituals create conditions under which common knowledge can be generated. This theory fits nicely with understanding common knowledge through the lens of evident knowledge events. Another potential direction for future research would be to try to mathematically model this ritualistic generation of common knowledge.

Finally, common knowledge, common  belief, and factional belief can be used as tools to help understand the formation and transformation of social norms. In particular, Cristina Bicchieri proposed and meticulously advocated for a definition of social norms of which our definition of factional belief is very reminiscent \cite{Bicchieri2005}.

\section{Defining Factional Belief}
\label{section:defining}
For the following definitions,  let ($\Omega$, $\Sigma$, $\Pr$) be a probability space, where $\Omega$ denotes a set of states, $\Sigma$ denotes a $\sigma$-algebra of events, and $\Pr$ denotes a probability measure on $\Sigma$. Let $I$ denote a set of agents. 

For each $i \in I$, $\Pi_i$ is a partition of $\Omega$ into measurable sets with positive probability. It is, therefore, a countable partition. For $\omega \in \Omega$, the element of $\Pi_i$ that contains $\omega$ is written as $\Pi_i(\omega)$. $\Pi_i$ can be interpreted as the information available to agent $i$. That is, $\Pi_i(\omega)$ is the set of states that are indistinguishable to $i$ when $i$ observes $\omega$. Let $B_i^{p}(E)$ denote the event that agent $i$ believes in event $E$ with probability at least $p$. Formally, we write $B_i^{p}(E) = \set{\omega : \Pr[E | \Pi_i(\omega)] \geq p}$. Lastly, for events $E$ and $F$, we use the notation $E \sse F$ to denote that, whenever $E$ occurs, $F$ occurs.\footnote{Note that, unlike with knowing, it is possible for an agent to believe something that is not true. In particular, $B_i^{p}(E)$ need not be a subset of $E$.}

The following examples are helpful to illustrate this notation. When rolling a fair die, with  equally probably outcomes in the set $\set{1, 2, 3, 4, 5, 6}$ we have $\set{2, 4} \sse \set{\text{Outcome is even}}$. Similarly, when the die has been tossed, but the outcome has not been revealed, we have the event $B_i^{\frac{1}{2}}( \set{\text{Outcome is even}})$ for any agent $i$. 

This notation is borrowed from Monderer and Samet \cite{Monderer1989}, and the rest of the terms and claims in this section are defined and stated, respectively, to be analogous to those from their paper.

\begin{definition} [Evident ($p$, $\mu$)-belief]
\label{def:evidentpm}
An event $E$ is an evident $(p, \ \mu)$-belief if there exists (at least) a $\mu$ fraction such that whenever $E$ occurs, those agents assign a probability of at least $p$ to its occurrence. That is: 
\[
\exists \, J \sse I \text{ with $|J| \geq \mu |I|$ such that for each $j \in J$, } E \sse B_j^{p}(E).
\]
\end{definition}

Following, Monderer and Samet, we first define our notion of factional belief in terms of events that are evident ($p$, $\mu$)-belief. To maintain consistency with their work, we refer to this notion of factional belief as common ($p$, $\mu$)-belief.

\begin{definition}[Common ($p$, $\mu$)-belief]
\label{def:commonpm}
An event $F$ is common $(p, \ \mu)$-belief at $\omega \in \Omega$ if there exists an evident $(p, \ \mu)$-belief event $E$ such that $\omega \in E$ and 
\[
    \exists \, J \sse I \text{ with $|J| \geq \mu |I|$ such that for each $j \in J$, } E \subseteq B_j^{p}(F).
\]
\end{definition}

That is, $F$ is common $(p, \ \mu)$-belief whenever there is an event ($E$) that is an evident $(p, \ \mu)$-belief whose occurrence implies the existence of (at least) a $\mu$ fraction of agents that believe with probability at least $p$ in $F$. Note that any event $E$ that is an evident $(p, \ \mu)$-belief is trivially also common $(p, \ \mu)$-belief (with $F = E$).

Now, an important property shared by common knowledge and common belief is the formal equivalence of their definitions in terms of evident events and their intuitive definitions as infinite hierarchies. Common $(p, \ \mu)$-belief retains this property. In order to state this result formally in Proposition~\ref{prop:one}, we need to formally define the infinite hierarchy, which is done below in Definition~\ref{def:hierarchy}. 

Informally, each level ($n \geq 1$) in this hierarchy, refers to the event that there exists (at least) a $\mu$ fraction of agents who believe with probability (at least) $p$ in the previous level of the hierarchy. The initial level ($n = 0$) is simply the relevant event $F$. So, written out entirely, the full informal definition would be that an event $F$ is common $(p, \ \mu)$-belief if there exists a $\mu$ fraction of agents who believe $F$ with probability $p$, there exists a $\mu$ fraction of agents who believe with probability $p$ that there exists a $\mu$ fraction of agents who believe $F$ with probability $p$, and so on, ad infinitum.

\begin{definition}
\label{def:hierarchy}
For every event $F$ and every $0 \leq p \leq 1$ let 
\[
    E^{p, \mu}(F) = \bigcap_{n \geq 1} F_{\mu}^n,
\] 
where $F_{\mu}^0 = F$ and $F_{\mu}^n$ is the event ``$\exists \, J \sse I \text{ with $|J| \geq \mu |I|$ such that $\forall j \in J$, } B_{j}^p(F_{\mu}^{n-1})$''.
\end{definition}

\begin{proposition}
\label{prop:one}
For every event $F$, every $0 \leq p \leq 1$, and every $0 \leq \mu \leq 1$:
\begin{enumerate}
    \item $E^{p, \mu}(F)$ is an evident $(p, \ \mu)$-belief and $\exists \, J \sse I$  with $|J| \geq \mu |I|$ such that $\forall j \in J$, $E^{p, \mu}(F) \sse B_j^{p}(F)$.
    
    \vspace{1 ex}
    
    \item $F$ is common $(p, \ \mu)$-belief at $\omega$ if and only if $\omega \in E^{p, \mu}(F)$.
\end{enumerate}
\end{proposition}

The proof of this proposition is essentially the same as the proof of the analogous proposition (Proposition 2) in Monderer's and Samet's paper \cite{Monderer1989}. However, the details of the proof are not particularly relevant or instructive with regard to our contributions in this work, so we omit them here and consign them to Appendix~\ref{appendix:proof_prop2.4}. The important takeaway from this proposition, is, as noted above, the formal equivalence between the definition of common $(p, \ \mu)$-belief in terms of evident $(p, \ \mu)$-belief (Definition \ref{def:commonpm}) and the hierarchical definition (Definition \ref{def:hierarchy}). As with the analogous definitions for common knowledge and common belief, the latter definition is more intuitive and perhaps more natural, but the former definition is more mathematically convenient and is what we will reference in what follows. The former definition is also a notion that better corresponds to how agents might be expected to reason about this type of belief in reality, since it is unrealistic to suppose that they consider infinite hierarchies of beliefs.

\section{Model}
\label{section:model}
Let $G = (V, E)$ be a graph and $I$ be a set of $n$ agents, such that each vertex $v_i \in V$ corresponds to an agent $i \in I$. We will think of $G$ as representing a social network of strategic agents who are participating in a revolt game. The graph $G$ is common knowledge among the agents.

Nature draws the \text{state} of the world $s \in S$ according to a distribution $D_S$ and selects a \textit{type} $t_i \in T = \set{\alpha, \nu} \cup X$ for each agent $i$, where $X = \cup_{j} \set{\chi_j}$ is non-empty (and finite). The types are selected independently at random according to a distribution $D_T^{s}$ associated with state $s$. $D_S$ and $D_T^{s}$ for each $s \in S$ are common knowledge among the agents.

Each agent $i$ will observe the type of each agent $k \in I$ such that $(v_i, v_k) \in E$ (this set of agents constitutes the set of \textit{neighbors} of $i$). The information resulting from this observation---an agent's type and the types of all of her neighbors---defines that agent's \textit{context}. When agent $i$ has the context $c$, we write it as $c(i) = c$. We use $C$ to denote the set of all contexts that are possible in $G$.

\textit{Ex ante}, or, before selection of the state, the assignment of types, and the observation of contexts, agents choose a pure \textit{strategy} $\sigma: C \to \set{R, Y}$, where $\set{R, Y}$ is the set of \textit{actions}.\footnote{``\underline{R}evolt'' and ``\underline{Y}ield,'' respectively.} A \textit{strategy profile} $(\sigma_1, \sigma_2, \ldots, \sigma_n)$ is collection of strategies for each agent. Let $\sigma_{-i} = (\sigma_1, \sigma_2, \ldots, \sigma_{i - 1}, \sigma_{i + 1}, \ldots \sigma_n)$ denote the strategies of all the agents except for $i$. Similarly, \textit{ex post}, or, after the selection of types and observation of contexts, an \textit{action profile} $(a_1, a_2, \ldots, a_n)$ where $a_i = \sigma_i(c(i))$ is a collection of actions for each agent.

Let $\mathcal{R}(a_1, a_2, \ldots, a_n) = \set{i \in [n] : a_i = R}$ denote the set of agents who play the action $R$ (revolt) given their context and strategy. Agent $i$ with type $t_i = t$ receives a payoff according to the function $f_i^t: \set{R, Y}^n \to [0, 1]$:

\begin{align*}
    & f_i^{\alpha}(a_1, a_2, \ldots, a_n) = 
    \begin{cases}
    1, & \text{if } a_i = R,\\
    0, & \text{otherwise},
    \end{cases}\\
    & f_i^{\nu}(a_1, a_2, \ldots, a_n) = 
    \begin{cases}
    1, & \text{if } a_i = Y,\\
    0, & \text{otherwise},
    \end{cases}\\
    &\text{and}\\
    & f_i^{\chi_j}(a_1, a_2, \ldots, a_n) = 
    \begin{cases}
    1 - p_j, & \text{if } |\mathcal{R}(a_1, a_2, \ldots, a_n)| \geq \mu_j \cdot n \text{ and } a_i = R,\\
    p_j, & \text{if } |\mathcal{R}(a_1, a_2, \ldots, a_n)| < \mu_j \cdot n \text{ and } a_i = Y,\\
    0, & \text{otherwise}.
    \end{cases}
\end{align*}
where $p_j, \mu_j \in [0, 1]$ for each $j$ and are common knowledge among the agents.

We call $\sigma_i$ a \textit{best-response} to $\sigma_{-i}$ when $\sigma_i$ maximizes agent $i$'s \textit{ex ante} expected payoff given $\sigma_{-i}$, which, by linearity of expectation, is equivalent to maximizing her expected payoff for each $c \in C$ (using the beliefs she would have about the contexts of the other agents after observing $c$). 

We say that a strategy profile $(\sigma_1, \sigma_2, \ldots, \sigma_n)$ is an \textit{equilibrium} when each strategy $\sigma_i$ is a best-response to $\sigma_{-i}$.  

Intuitively, the game is defined so that there is a natural strategy for each agent based on type:
\begin{itemize}
    \item Agents of type $\alpha$ should \underline{a}lways revolt. 
    
    \item Agents of type $\nu$ should \underline{n}ever revolt. 
    
    \item Agents of type $\chi_j$ should \underline{c}onditionally revolt: 
    \begin{itemize}
        \item Each type $\chi_j$ is characterized by a pair of thresholds $(p_j, \ \mu_j)$ that indicate an agent of type $\chi_j$ should revolt when she believes that $\Pr[|\mathcal{R}(a_1,a_2,\ldots,a_n)| \geq \mu_j] \geq p_j$. 
    \end{itemize}
\end{itemize}

An important thing to note about this game is that when there are agents of type $\nu$, there is not an equilibrium where each agent revolts regardless of her beliefs. Similarly, when there are agents of type $\alpha$, there is not an equilibrium where each agent chooses not to revolt regardless of her beliefs. 

Even without these, though, there are still potentially multiple equilibria corresponding to revolts of different sizes, and our work does not address equilibrium selection. Consequently, in what follows, we will write that a revolt (of a particular size) is \textit{supported} in equilibrium instead of writing that a revolt will occur. Similarly, we write that agents are \textit{secure enough to revolt} when they sufficiently believe their thresholds are met instead of writing that agents will revolt. 

Lastly, in this work we primarily consider the \textit{largest} revolts that are supported in some equilibrium. Such revolts are supported by \textit{symmetric} equilibria, in which each agent adopts the same strategy---namely, the strategy detailed in the bullet points above.

\subsection{Motivating Example}
To see this model in practice and get a feeling for how agents need to reason about the concepts that we have introduced, we will work through a modest example.

Suppose that $G$ is a grid of $n$ vertices embedded on a torus, so that each the vertex associated with each agent is adjacent to the vertices representing exactly four other agents (and there is no boundary). Thus, a context in this graph consists of the central agent and her type plus the types of each of her neighbors in $G$. 
In this example, we will have two types of agents: (1) Agents of type $\chi$, who want to revolt conditionally; they feel secure enough to revolt when the threshold pair $(p = \frac{2}{5}, \ \mu = \frac{1}{2})$ is satisfied. (2) Agents of type $\nu$, who never want to revolt. We will also have two equally-likely states: an anti-government state $A$, where agents are of type $\chi$ with probability $\frac{4}{5}$ and a pro-government state $B$ where agents are of type $\nu$ with probability $\frac{4}{5}$.

We also define the notion of a \textit{candidate agent}. In this setting, for reasons that will become clear in the proof of the following proposition, we refer to agents of type $\chi$ with two or more neighbors of type $\chi$ as \textit{candidate agents}.

\begin{proposition}
In this model, when $G$ is sufficiently large, the event that at least $\frac{1}{2}$ of the agents are candidate agents is an evident $(\frac{2}{5}, \ \frac{1}{2})$-belief. When it occurs, it is common $(\frac{2}{5}, \ \frac{1}{2})$-belief that supports a revolt of size $\frac{1}{2}$.
\end{proposition}
\begin{proof}
First, we will compute $\Pr[\text{State} = A | c]$ for each context $c$ and use the computed values to demonstrate that our definition of candidate agents is the correct one for this setting; candidates are likely to feel secure enough to revolt, based on their $p$-thresholds.

Each agent has 5 independent samples from the probability distribution defined by the state, with which to calculate the probability of each state, given her information. The first step in this is to compute the likelihood of her context given the state, which can be calculated by evaluating the probability mass function of a binomial distribution with the appropriate parameters. Then, the probability for each state given a certain neighborhood can be calculated using Bayes' rule. The results of these calculations are shown in Tables \ref{tab:binom} and \ref{tab:bayes}. 
    
Now, agents of type $\chi$ that only have at most one other neighbor of type $\chi$, believe that the state is $A$ with probability at most $\frac{1}{5}$, so they will not feel secure enough to revolt; a revolt of size $\mu = \frac{1}{2}$ is very unlikely to be supported when the state is $B$, and they do not sufficiently believe that the state is $A$. 
    
As a result, it is exactly agents of type $\chi$ with contexts such that they have two or more neighbors of type $\chi$, which we defined earlier as candidate agents, in which we are interested (cases with $k \geq 3$ in Table~\ref{tab:bayes}). In particular, we are interested in the probability that a majority of agents are candidates. 
\begin{claim*}
The probability that at least half the agents are candidates, given that the state is $A$ is at least $\frac{1739}{3125}$.
\end{claim*} 
\begin{claimproof}[Proof of Claim]
For this, we can try to count all the non-candidate agents and see how likely they are to outnumber the candidates. 

Let $X$ count the number of non-candidate agents. Specifically, let $X = \sum_{i \in [n]} X_i$, where $X_i$ is an indicator variable that is 1 if and only if agent $i$ is not a candidate. For each $X_i$, therefore, the expectation of $X_i$ is equal to the probability that $i$ is not a candidate, which can happen in three ways: (1) $i$ is of type $\nu$, (2) $i$ is of type $\chi$ and has no neighbors of type $\chi$, or (3) $i$ is of type $\chi$ and has one neighbor of type $\chi$. The probability of (1) is trivially $\frac{1}{5}$. Conditioned on $i$ being of type $\chi$, the respective probabilities of the remaining possibilities are (2) $\frac{1}{625}$ and (3) $\frac{16}{625}$. To obtain unconditional probabilities for (2) and (3), we multiply each conditional probability by $\frac{4}{5}$. Consequently, since the events (1), (2), and (3) are disjoint, we have the following:
    \[
        \E[X_i] = \frac{1}{5} + \frac{4}{5} \left( \frac{1}{625} \right) + \frac{4}{5} \left( \frac{16}{625} \right) = \frac{693}{3125}.
    \]
Therefore, by  linearity of expectation, $\E[X] = \frac{693n}{3125}$.

Now, we need to upper bound the probability that the number of non-candidate agents is the majority. Ideally for this, we would like to use something tight, like a typical Chernoff-Hoeffding bound. Unfortunately, here, the $X_i$'s are not independent; the contexts of agents that share neighbors are correlated, which complicates the calculation.

For our more rigorous analysis later, we will insist on tighter probability bounds. But for the sake of simplicity, in this example, Markov's inequality is sufficient:
    \[
        \Pr\left[X \geq \frac{n}{2}\right] \leq \frac{\frac{693n}{3125}}{\frac{n}{2}} = \frac{1386}{3125} \approx 0.44.
    \]
Consequently, the probability that at least half the agents are candidates, given that the state is $A$, is at least $(1 - \frac{1386}{3125}) = \frac{1739}{3125}$.
\end{claimproof}

Candidate agents can perform this exact same calculation and, further, they believe that the state is $A$ with probability at least $\frac{4}{5}$. Thus, when the graph is large enough that the context of any individual agent is inconsequential in their reasoning about the total fraction of candidate agents, the probability that they assign to at least half of the agents being candidate agents is at least $\frac{4}{5} \cdot (\frac{1739}{3125}) \geq \frac{2}{5}$.

That is, in sufficiently large graphs, candidate agents believe with probability at least $\frac{2}{5}$ that at least $\frac{1}{2}$ of the agents are candidate agents. As a result, the event that at least $\frac{1}{2}$ of the agents are candidate agents is an evident $(\frac{2}{5}, \ \frac{1}{2})$-belief. Consequently, when it occurs, by definition, it is common $(\frac{2}{5}, \ \frac{1}{2})$-belief. In this event, a revolt of size $\frac{1}{2}$ is supported, since at least half of the agents have their thresholds satisfied, and consequently feel secure enough to revolt.
\end{proof}

\begin{table}[t]
\caption{Results of the probabilistic calculations for the motivating example.}

\begin{subtable}[t]{0.49\textwidth}
  \centering
   \caption{The likelihood, in each state, of having a context with $k$ agents of type $\chi$.}
   {\renewcommand{\arraystretch}{1.25}
\begin{tabular}{ |c|c|c| } 
 \hline
  & $\Pr[ c | \text{State}  = A]$ & $\Pr[ c | \text{State} = B]$ \\
  \hline
 $k = 0$ & $\frac{1}{3125}$ & $\frac{1024}{3125}$ \\
 \hline
 $k = 1$ & $\frac{4}{625}$ & $\frac{256}{625}$ \\ 
 \hline
 $k = 2$ & $\frac{32}{625}$ & $\frac{128}{625}$ \\ 
 \hline
 $k = 3$ & $\frac{128}{625}$ & $\frac{32}{625}$ \\ 
 \hline
 $k = 4$ & $\frac{256}{625}$ & $\frac{4}{625}$ \\
 \hline
 $k = 5$ & $\frac{1024}{3125}$ & $\frac{1}{3125}$ \\
 \hline
\end{tabular}
}
\label{tab:binom}
\end{subtable}
\begin{subtable}[t]{0.49\textwidth}
  \centering
   \caption{The likelihood of state $A$, given a context with $x$ agents of type $\chi$.}
   {\renewcommand{\arraystretch}{1.25}
\begin{tabular}{ |c|c| } 
 \hline
  & $\Pr[\text{State} = A | k = x]$ \\
  \hline
 $x = 0$ & $\frac{1}{1025}$ \\
 \hline
 $x = 1$ & $\frac{1}{65}$  \\ 
 \hline
 $x = 2$ & $\frac{1}{5}$  \\ 
 \hline
 $x = 3$ & $\frac{4}{5}$ \\ 
 \hline
 $x = 4$ & $\frac{64}{65}$ \\
 \hline
 $x = 5$ & $\frac{1024}{1025}$ \\
 \hline
\end{tabular}
}
 \label{tab:bayes}
\end{subtable}

\label{tab:calculations}
\end{table}

\section{Applying Factional Belief: A Computational Perspective}
\label{section:applying}
The network revolt game model described in Section 3 gives us a useful setting in which to apply our definitions from Section 2 and demonstrate how they can provide insight into strategic coordination. In this section, we explore the interaction between factional belief and strategic coordination from a computational perspective, by trying to answer an elementary question: When can we efficiently determine if strategic coordination (i.e. in our model, a revolt) is supported under given conditions?

In pursuing an answer to this question, we seek to apply the intuition that we have constructed in our motivating example to a more general setting. Toward that end, a relatively modest generalization of the model used in the example---our Fundamental Case, below---is rich enough to provide an interesting, non-trivial answer to our question and to provide robust intuition that guides us through the various extensions of the model that we discuss in Section~\ref{section:broadening}.

\subsection{Fundamental Case: Low-Degree Graphs, Two States, and Three Types}
\label{subsection:fundamental-case}
Recall that our model requires a description of possible agent types and possible states (which specify probability distributions over those types). For our initial case, there are three agent types and two possible states.

Our focus will be on agents who want to revolt conditionally---agents of type $\chi$---who have the threshold pair $(p, \ \mu)$ for an arbitrary $0 \leq p \leq 1$ and $0 \leq \mu \leq 1$. In some sense, these are the truly ``strategic'' agents; they need to coordinate with other agents in order to feel secure enough to revolt. There are also agents who always behave in a prescribed manner, regardless of their contexts or the state: pro-government agents of type $\nu$, who will never revolt, and anti-government agents of type $\alpha$, who will always revolt.

The possible states are $A$, an anti-government state, and $B$, a pro-government state. $D_S$, the distribution over the states, $D_T^A$ and $D_T^B$, the distributions over the types in each state, and $p$ and $\mu$, the threshold values for agents of type $\chi$, are commonly known to the agents under the prior $P$. Intuitively, the labels assigned to the states as being anti- and pro-government correspond to an implication that as the number of agents grows, the size of the largest revolt supported in state $A$ should be larger than in state $B$. In what follows, we assume that the given labels are assigned correctly, but in practice the correct labels can be determined by running an algorithm that we present later with each possible labeling and comparing the results. When the states are correctly labeled, Algorithm~\ref{alg:one} will return $X_A$ and $X_B$ such that $X_A \geq X_B$. 

 In this more concrete setting, we are able to pose a straightforward question: Given values $\mu^*$ and $q^*$, is a revolt of size (at least) $\mu^*$ supported with probability at least $q^*$? We refer to this problem as \textsc{Revolt}.

\begin{definition}[{\textsc{Revolt}}] \hfill

\textbf{Given:} $(G, P, \mu^*, q^*)$, where $G$ is a graph with $n$ vertices and $P = (p, \mu, D_T^A, D_T^B, D_S)$ is the common prior. 

\textbf{Question:} Is a revolt of size at least $\mu^*$ supported with probability at least $q^*$? 
\end{definition}

Note that there is a nuance regarding the timing of the network revolt game model described in Section 3. \textsc{Revolt} considers the likelihood of revolt of a certain size being supported in a given network \textit{ex ante}---i.e. before the selection of the state and the assignment of types to agents. 

Although the question posed by \textsc{Revolt} is straightforward to state, it is not straightforward to solve efficiently. In fact, we have the following hardness result, which we prove in Appendix~\ref{appendix:hardness}:

\begin{proposition}
\label{prop:hardness}
\textsc{Revolt} is NP-hard.
\end{proposition}

Consequently, we will instead consider a relaxed version of the problem that we will be able to solve efficiently. In order to define this new problem, we first require $p$ and $\mu$ to be strictly between 0 and 1. We will also need two error terms, which will be constants given as input: $\epsilon, \delta > 0$. Lastly, we introduce two additional priors, as follows (recall $P = (p, \mu, D_T^A, D_T^B, D_S)$ is the common prior given as input to both problems):
\begin{align*}
    P^- & = (p - \delta, \mu - \epsilon, D_T^A, D_T^B, D_S),\\
    P^+ & = (p + \delta, \mu + \epsilon, D_T^A, D_T^B, D_S).
\end{align*}

Now, we are ready to define our new problem, \textsc{Promise Revolt}.

\begin{definition}[{\textsc{Promise Revolt}}] \hfill

\textbf{Given:} $(G, P, \mu^*, \epsilon, \delta)$, where $G$ is a graph with $n$ vertices, $P = (p, \mu, D_T^A, D_T^B, D_S)$ is the common prior, and $\epsilon, \delta > 0$ are constants.

\textbf{Output:} When exactly one of the following cases is true, output the corresponding symbol:
\begin{itemize}
    \item[$\mathbf{\Omega}$:] A revolt of size $\mu^* - \epsilon$ is supported with probability $q^* \geq 1 - \delta$ under the prior $P^-$.
    
    \item[$\mathbf{A}$:] A revolt of size $\mu^*$ is supported with probability $q^* \in [\Pr[\text{State} = A] - \delta, \Pr[\text{State} = A] + \delta]$ under the prior $P$.
    
    \item[$\bm{\emptyset}$:] A revolt of size $\mu^* + \epsilon$ is supported with probability $q^* \leq \delta$ under the prior $P^+$.
\end{itemize}

\end{definition}

\textsc{Promise Revolt} is named so as to emphasize that this it is a promise problem, in the typical sense. The ``promise'' is that the given instance is such that exactly one of the cases from the definition of the problem is true. Given that promise, the three cases are mutually exclusive and any solution is required to always output the correct answer. It is exactly this promise to exclude difficult inputs that makes \textsc{Promise Revolt} easier to solve than \textsc{Revolt}.

Still, in providing an algorithm to solve \textsc{Promise Revolt}, we will see that the problem has a structure that allows us to closely approximate \textsc{Revolt} by solving \textsc{Promise Revolt} (for large graphs). This structure is discussed in more detail at the end of this subsection (4.1), particularly with respect to Figure~\ref{fig:equilibria}, which helps to illustrate this intuition.
 
 Lastly, note that we assume all inputs to \textsc{Revolt} and \textsc{Promise Revolt} are given as rational numbers. This leads us to state our first theorem:

\begin{theorem} 
\label{theorem:alg}
Given $\epsilon > 0$, the prior $P$, and $\mu^* \in [0, 1]$, there exists $\delta(n) \in \frac{1}{\exp \left( \Omega_{\epsilon, P} \left(\sqrt[3]{n} \right) \right)}$ such that for any graph $G$ with $n$ vertices where the largest degree of any vertex is $O_{\epsilon, P} \left( \sqrt[3]{n} \right)$, Algorithm~\ref{alg:three} can be used to solve \textsc{Promise Revolt}$(G, P, \mu^*, \epsilon, \delta = \delta(n))$ in polynomial time.
\end{theorem}

\begin{algorithm}
\SetAlgoNoLine
\KwIn{A graph $G = (V, E)$ of $n$ vertices, the prior $P = (p, \mu, D_T^A, D_T^B, D_S)$.}
\KwOut{$X_A$ and $X_B$, the maximum size of revolt supported in states $A$ and $B$, respectively.}
\vspace{1 ex}
Let $e_s(\tau)$ denote the expected fraction of type $\tau$ agents in state $s$.\\
Define the set of candidate states $S_C = \set{s \in \set{A, B} : e_s(\chi \cup \alpha) \geq \mu}$.\\
\If{$S_C = \emptyset$}{
    $X_s = e_s(\alpha) \  \forall s \in \set{A, B}$.
}
\If{$S_C = \set{A, B}$}{
    $X_s = e_s(\chi \cup \alpha) \  \forall s \in \set{A, B}$.
}
\If{$S_C = \set{A}$}{
    Let $C(\chi)$ be the set of contexts centered around agents of type $\chi$.\\
    Define the set of candidate contexts $C_C = \set{c \in C(\chi) : \Pr[\text{State} = A | c] \geq p}$.\\
    (Below, we slightly abuse notation, treating $C_C$ as if it were a type when writing $e_s(C_C)$.)\\
    \If{$e_A(C_C \cup \alpha) \geq \mu$}{
        $X_s = e_s(C_C \cup \alpha) \  \forall s \in \set{A, B}$.
    }
    \Else{
        $X_s = e_s(\alpha) \  \forall s \in \set{A, B}$.
    }
}
\Return $X_A$, $X_B$.

    \caption{Finding the size of largest revolt supported in each state.}
	\label{alg:one}
\end{algorithm}

\begin{algorithm}
\SetAlgoNoLine
\KwIn{$X_A$, $X_B$ (output from Algorithm~\ref{alg:one}) and $\mu^*$.}
\KwOut{$\mathbf{\Omega}$, $\mathbf{A}$, or $\bm{\emptyset}$.}
\vspace{1 ex}
\If{$X_A \geq \mu^*$ and $X_B \geq \mu^*$}{
    \Return $\mathbf{\Omega}$
}
\If{$X_A \geq \mu^*$ and $X_B < \mu^*$}{
    \Return $\mathbf{A}$
}
\If{$X_A < \mu^*$ and $X_B < \mu^*$}{
    \Return $\bm{\emptyset}$
}

\caption{Comparing the size of the largest revolt supported in each state.}
\label{alg:two}
\end{algorithm}

\begin{algorithm}
\SetAlgoNoLine
\KwIn{$(G, P, \mu^*, \epsilon, \delta)$, where $G$ is a graph of $n$ vertices, $P = (p, \mu, D_T^A, D_T^B, D_S)$ is the common prior, and $\epsilon, \delta > 0$ are constants}
\KwOut{$\mathbf{\Omega}$, $\mathbf{A}$, $\bm{\emptyset}$, or \textbf{\textsc{Null}}.}
\vspace{1 ex}
$X_{A_1}, X_{B_1} \gets \textbf{ALGORITHM~\ref{alg:one}}(G, P = (p', \mu', D_T^A. D_T^B, D_S))$, where $p' = p + \frac{\delta}{3}$, $\mu' = \mu + \frac{\epsilon}{3}$.\\
$S_1 \gets  \textbf{ALGORITHM~\ref{alg:two}}(X_{A_1}, X_{B_1}, \mu^*)$.

\vspace{1 em}

$X_{A_2}, X_{B_2} \gets \textbf{ALGORITHM~\ref{alg:one}}(G, P = (p', \mu', D_T^A. D_T^B, D_S))$, where $p' = p - \frac{\delta}{3}$, $\mu' = \mu - \frac{\epsilon}{3}$.\\
$S_2 \gets  \textbf{ALGORITHM~\ref{alg:two}}(X_{A_2}, X_{B_2}, \mu^*)$.

\If{$S_1 = S_2$}{
    \Return $S_1$
}
\Else{
    \Return \textbf{\textsc{Null}}
}

\caption{Solving \textsc{Promise Revolt}}
\label{alg:three}
\end{algorithm}

Before we get to the proof of Theorem~\ref{theorem:alg}, we briefly discuss our assumption that each agent has a degree that is $O\left(\sqrt[3]{n}\right)$. Primarily, this assumption serves to simplify our analysis. For simplicity, it is convenient to make some distinction between ``low-degree'' agents and ``high-degree'' agents to highlight the fact that a prototypical high-degree agent would have a lot more information about the state than a prototypical low-degree agent. However, any particular choice of cutoff to separate high- and low-degree agents is somewhat arbitrary. Here, we choose $O\left(\sqrt[3]{n}\right)$, which is mathematically convenient for defining an upper bound on our error function $\delta(n)$ in Theorem~\ref{theorem:alg}. 

Given this cutoff, we focus first on the case where there are only low-degree agents, which is sufficient to provide the guiding intuition that we will follow in the next section, when we discuss broadening the setting in various ways. One such extension will involve allowing vertices of arbitrary degree.

We now proceed by making several arguments which form the building blocks of the proof of Theorem~\ref{theorem:alg} that follows.

\begin{lemma}
\label{lemma:poly}
Algorithms~\ref{alg:one} and \ref{alg:two} terminate in polynomial time with respect to the number of agents.  
\end{lemma}

\begin{proof}
The key for this property of Algorithm~\ref{alg:one} is that contexts are identity-agnostic, so the number of contexts is polynomial in $n$ when the number of types is a constant. Therefore, we are able to enumerate all of the possible contexts in polynomial time. For each of these contexts $c$, we can compute the relevant probabilities---$\Pr[\text{State} = s | c]$ and $\Pr[c | \text{State} = s]$ for both states $s$---using Bayes' rule. 
The rest of the steps in the algorithm are linear in the number of contexts, and therefore also computable in polynomial time. Algorithm~\ref{alg:two} is trivially computable in constant-time.
\end{proof}

Having shown that they are polynomial-time computable, we now demonstrate the correctness of Algorithms~\ref{alg:one} and \ref{alg:two} under idealized conditions.

\begin{claim*}
When agents \underline{believe} with probability $1$ that the actual size of the largest supported revolt in each state will exactly equal its expected size, then Algorithm~\ref{alg:one} correctly computes the expected size of the largest revolt in each state. When, further, it is \underline{true} that the actual size of the largest supported revolt in each state will exactly equal its expected size, then Algorithm~\ref{alg:two} identifies the set of states in which revolt of size $\mu^*$ is supported with no error.
\end{claim*}

\begin{claimproof}[Proof of Claim]
Algorithm~\ref{alg:one} first defines the set of candidate states---these are the states in which it is possible, but not necessarily the case, that type-$\chi$ agents will feel secure enough to revolt, because there are at least $\mu$ $\alpha$- and $\chi$-type agents. If there are no such states, then only agents of type $\alpha$ will revolt. If both states are candidates, then all $\alpha$- and $\chi$-type agents will feel secure enough to revolt: Agents of type $\alpha$ always revolt and agents of type-$\chi$ have their $\mu$ threshold met in both states, by the definition of a candidate state. As a result, their $p$ threshold is also necessarily met, because the probability of the state being either $A$ or $B$ is $1 \geq p$. The most interesting case is when only $A$ is a candidate state. In this case, only type $\chi$ agents who $p$-believe that the state is $A$ will have their $p$-threshold met. So, similarly to our example from Section 3.1, we call those agents candidate agents (and refer to their contexts as candidate contexts.) If the number of candidate agents and $\alpha$-type agents, given that the state is $A$, is at least $\mu$, then all of those agents feel secure enough to revolt. This is true in any state, because even when the state is $B$, candidate agents, by definition, $p$-believe that the state is $A$. 

Algorithm~\ref{alg:two} simply compares the size of the expected revolt in each state to $\mu^*$ to decide in which states, if any, a revolt of expected size at least $\mu^*$ is supported. If the actual size of the revolt in any state is exactly equal to its expectation, as we assume, then Algorithm~\ref{alg:two} introduces no error.
\end{claimproof}

Now, the assumption that the actual supported revolt exactly equals the expected size of the supported revolt is, of course, too strict. However, we will be able to show that the actual size of the supported revolt concentrates around its expectation, and consequently, the errors in our algorithms decrease quickly as $n$ grows.

Recall that computing the expected size of a supported revolt involves counting the number of some subset of three kinds of agents: $\alpha$-type agents, $\chi$-type agents, and candidate agents. For $\alpha$-type and $\chi$-type agents, it is simple to show that the number of such agents concentrates around its expectation because agent types are assigned independently at random.

\begin{lemma}
\label{lemma:concentration-alpha-chi}
Given $\epsilon > 0$, the probability that, in a given state, the number of agents of type $\chi$ and of type $\alpha$ differ from the expected number of agents of type $\chi$ and of type $\alpha$ by a multiplicative factor of $\epsilon$ is at most $\delta =\delta(n)$ for some $\delta(n) \in \frac{1}{\exp \left( \Omega_{\epsilon, P} \left(\sqrt[3]{n} \right) \right)}$.
\end{lemma}
\begin{proof}
Because the agent types are drawn independently at random, we can use a standard Chernoff-Hoeffding bound on the expected number of each type of agents. Consequently, the difference between the actual and expected number of each type of agent decays exponentially with $n$, and therefore we can choose $\delta(n) \in \frac{1}{\exp \left( \Omega_{\epsilon, P} \left(\sqrt[3]{n} \right) \right)}$ so that the difference is trivially smaller. 
\end{proof}

Counting the number of candidate agents, as foreshadowed in the Motivating Example of Section 3.1, is somewhat more complicated. Here, we need to consider the contexts of agents---not just their type---and agents' contexts are correlated with the contexts of their neighbors and their neighbors' neighbors in $G$. We will show, though, that as the graph grows, the effect of this correlation is small. In fact, we will still be able to prove an exponentially-decreasing bound on the difference between the actual and expected number of candidate agents.
\begin{lemma}
\label{lemma:concentration-candidate}
Given $\epsilon > 0$, the probability that, in a given state, the number of candidate agents is less than the expected number of candidate agents by a multiplicative factor of $\epsilon$ is at most $\delta =\delta(n)$ for some $\delta(n) \in \frac{1}{\exp \left( \Omega_{\epsilon, P} \left(\sqrt[3]{n} \right) \right)}$.
\end{lemma}
\begin{proof}
Let $X_{C}$ be a random variable that denotes the number of candidate agents. We can write $X_{C} = \sum_{i = 1}^n X_{i, C}$ where $X_{i, C}$ is an indicator variable that indicates whether or not agent $i$ is a candidate. 

As we have noted above, an agent's status as a candidate is dependent on her context, and as a result, the variables $X_{i, C}$ are not independent. However---given our assumption that the maximum degree of any vertex in $G$ is $O\left(\sqrt[3]{n}\right)$---their dependence is constrained enough that we are able to apply a useful exponential bound involving the \textit{fractional chromatic number} $\chi^*(\Gamma)$ of the constraint graph $\Gamma$ of the random variables $X_{i, C}$. 

(See Theorem~\ref{theorem:dependent_bound} in  Appendix~\ref{appendix:dependent_bound} for additional details).

For our purposes, it is sufficient to use a trivial bound on the fractional chromatic number of any graph. The fractional chromatic number of a graph is at most the \textit{chromatic number} of the graph, which is at most the maximum degree of any vertex plus one. So in $\Gamma$, where edges correspond to dependencies between pairs of random variables ($X_{i, C}, X_{j, C}$), the maximum degree of any vertex---and therefore the fractional chromatic number $\chi^*(\Gamma)$---is $O_{\epsilon, P}\left(n^{\frac{2}{3}}\right)$.

\begin{align*}
    \Pr[|X_{C} - \E[X_{C}]| \geq \epsilon n] \leq 2 \exp \left(\frac{-2\epsilon^2 n}{\chi^*(\Gamma)}\right) \in \frac{1}{\exp \left( \Omega_{\epsilon, P} \left(\sqrt[3]{n} \right) \right)}
\end{align*}
\end{proof}

\begin{lemma} \hfill
\label{lemma:cases}
\begin{enumerate}
    \item Suppose Algorithm~\ref{alg:one} is run with the given inputs, but with modified $\mu' = \mu + \frac{\epsilon}{3}$ and $p' = p + \frac{\delta}{3}$, and Algorithm~\ref{alg:two} is run with the resulting $X_A$ and $X_B$ and the given $\mu^*$. If the result is $\mathbf{\Omega}$, then a revolt of size $\mu^* - \epsilon$ is supported with probability at least $1 - \delta$ under the prior $P$. If the result is $\mathbf{A}$, then a revolt of size $\mu^* - \epsilon$ is supported with probability at least $\Pr[\text{State} = A] - \delta$ under the prior $P$.
    
    \item On the other hand, suppose Algorithm~\ref{alg:one} is run with the given inputs, with modified $\mu' = \mu - \frac{\epsilon}{3}$ and $p' = p - \frac{\delta}{3}$, and Algorithm~\ref{alg:two} is run with the resulting $X_A$ and $X_B$ and the given $\mu^*$. If the result is $\bm{\emptyset}$, then a revolt of size $\mu^* + \epsilon$ is supported with probability at most $\delta$ under the prior $P$. If the result is $\mathbf{A}$, then a revolt of size $\mu^* + \epsilon$ is supported with probability at most $\Pr[\text{State} = A] + \delta$ under the prior $P$.
\end{enumerate}
\end{lemma}
\begin{proof}
This proof decomposes into two analogous arguments about (1) and (2), which themselves each contain two analogous arguments. We include only the first one in detail here, below, and claim the rest via analogy.
\begin{claim*}
Suppose that, in instance (1), the result is $\mathbf{\Omega}$. Then, the probability that a revolt of size $\mu^* - \epsilon$ is supported in both states is at least $1 - \frac{1}{\exp \left( \Omega \left(\sqrt[3]{n} \right) \right)}.$
\end{claim*}

\begin{claimproof}[Proof of Claim]
There are three cases for the nature of this revolt. 

\begin{enumerate}
    \item[(I)] The revolt consists of all $\alpha$-type agents.
    
    By Lemma~\ref{lemma:concentration-alpha-chi}, we know that, in either state, the probability that the expected number of $\alpha$-type agents differs from the expected number of $\alpha$-type agents by a multiplicative factor of $\epsilon$ is $\frac{1}{\exp \left( \Omega \left(\sqrt[3]{n} \right) \right)}$. Therefore, the probability that a revolt of size $\mu^* - \epsilon$ is supported in both states is at least $1 - \frac{1}{\exp \left( \Omega \left(\sqrt[3]{n} \right) \right)}.$
    
    \item[(II)] The revolt consists of all $\alpha$- and all $\chi$-type agents.
    
    In this case, the presence of $\chi$-type agents slightly complicates the analysis. Actual $\chi$-type agents have slightly easier thresholds to satisfy ($p$ and $\mu$) than the $\chi$-type agents considered by the algorithms ($p'$ and $\mu'$). Our algorithms determined that, in expectation, agents with $p'$ and $\mu'$ thresholds would feel secure enough to revolt. Each actual $\chi$-type agent, then, must also feel secure enough to revolt, not just in expectation: By Lemma~\ref{lemma:concentration-alpha-chi}, the probability that the actual number each of $\alpha$-type and $\chi$-type agents differs from its respective expectation by a multiplicative factor of $\frac{\epsilon}{6}$ is $\frac{1}{\exp \left( \Omega \left(\sqrt[3]{n} \right) \right)}$. Combining these, then, $\chi$-type agents believe with probability $1 - \frac{1}{\exp \left( \Omega \left(\sqrt[3]{n} \right) \right)} \geq p$ that at least $\mu' - \frac{\epsilon}{3} = \mu$ agents feel secure enough to revolt, and are thus themselves secure enough to revolt. Applying Lemma~\ref{lemma:concentration-alpha-chi} (from the perspective of Algorithm~\ref{alg:two} this time, not the perspective of $\chi$-type agents) we again incur two error terms of $\frac{\epsilon}{6}$ (one each for the $\chi$- and $\alpha$-type agents). Thus, we can conclude the probability that a revolt of size $\mu^* - \epsilon$ is supported in both states is at least $1 - \frac{1}{\exp \left( \Omega \left(\sqrt[3]{n} \right) \right)}.$ 
    
    \item[(III)] The revolt consists of all $\alpha$-type agents and all candidate agents.
    
    In the final case, we consider candidate agents instead of all agents of type $\chi$, and proceed exactly as we did in the second case, using Lemma~\ref{lemma:concentration-alpha-chi} to conclude that the number of $\alpha$-type agents concentrate and Lemma~\ref{lemma:concentration-candidate} to conclude that the number of candidate agents concentrates. Once again, the result is that the probability that a revolt of size $\mu^* - \epsilon$ is supported in both states is at least $1 - \frac{1}{\exp \left( \Omega \left(\sqrt[3]{n} \right) \right)}$.
\end{enumerate}  
\end{claimproof}

\begin{claim*}
Suppose that, in instance (1), the result is $\mathbf{A}$. Then, the probability that a revolt of size $\mu^* - \epsilon$ is supported in both states is at least $\Pr[\textup{State} = A] - \frac{1}{\exp \left( \Omega \left(\sqrt[3]{n} \right) \right)}.$
\end{claim*}

\begin{claimproof}[Proof of Claim]
The proof of this claim is analogous to the previous proof, where events that are assigned probability at least $1 - \frac{1}{\exp \left( \Omega \left(\sqrt[3]{n} \right) \right)}$ are instead assigned probability at least $\Pr[\text{State} = A] - \frac{1}{\exp \left( \Omega \left(\sqrt[3]{n} \right) \right)}$.
\end{claimproof}
 
The proof of the analogous claims for (2) are themselves analogous to the above cases for (1).
 
Finally, we note here that when $\chi$-type agents decide whether or not their thresholds are satisfied, they are not solely relying on their estimates of the fractions of different kinds of agents, as we describe above. They have additional knowledge, since they see the realized types of the agents in their context. However, for small graphs, $\delta(n)$ can be chosen to account for this. As the graph grows, since each agent has at most $O\left(\sqrt[3]{n}\right)$ neighbors, the consequences of observing the types of a few adjacent agents is negligible after conditioning on the state. As a result, for large enough $n$, after the agent reasons about the state, the probabilistic effect of the agent viewing the types in her context is subsumed by the $\epsilon$ error term. Furthermore, the appropriate choice of $\delta(n)$ also accounts for the $\frac{1}{\exp \left( \Omega \left(\sqrt[3]{n} \right) \right)}$ terms present in each of the claims stated above, for each $n$.
\end{proof}

\begin{proof}[Proof of Theorem~\ref{theorem:alg}]
It follows from Lemma~\ref{lemma:poly} that Algorithm~\ref{alg:three} terminates in polynomial time with respect to the inputs $G$ and $P$. We also note that Algorithm~\ref{alg:three} terminates in polynomial time with respect to $\frac{1}{\epsilon}$ and $\frac{1}{\delta}$. 

Further, it follows from Lemma~\ref{lemma:cases} that when Algorithm~\ref{alg:three} outputs a case, that case is always true. It only remains to show that when exactly one case in the statement of \textsc{Promise Revolt} is true, then Algorithm~\ref{alg:three} necessarily outputs that case: 

Here, the key is our use of three different potential revolt sizes ($\mu^* - \epsilon$, $\mu^*$, and $\mu^*+\epsilon$) and 3 different priors ($P^-$, $P$, and $P^+$) in defining the three cases of \textsc{Promise Revolt}. By promising that exactly one of those three cases is true, we guarantee that the values of $\mu$ and $p$ are sufficiently far---distance at least $\epsilon$ for $\mu$ and distance at least $\delta$ for $p$---from the crucial decision thresholds in Algorithm~\ref{alg:one} (e.g. the value $e_B(\chi \cup \alpha)$, which is used to determine whether or not $B$ is a candidate state)\footnote{Whether or not a decision threshold is in the set of \textit{crucial} decision thresholds---the thresholds from which our promise guarantees $p$ and $\mu$ are sufficiently far---depends on $\mu^*$. In addition to $e_B(\chi \cup \alpha)$, it can also include $e_A(C_C \cup \alpha)$ and the value of $p$ below which $e_A(C_C \cup \alpha) \geq \mu$ holds for candidate contexts and above which $e_A(C_C \cup \alpha) < \mu$.}. This ensures that both calls to Algorithm~\ref{alg:one} will return the same values.

\vspace{1 ex}

We can illustrate this counterfactually:

Let $\mu = e_B(\chi \cup \alpha) + \frac{\epsilon}{2}$ with $e_B(\chi \cup \alpha) > \mu^* > e_B(C_C \cup \alpha)$ and let $e_A(C_C\cup \alpha) > \mu + \frac{\epsilon}{2}$. 

Then, $A$ and $B$ would both be candidate states under the prior $P^-$. As a result, Algorithm~\ref{alg:one} (run with prior $P^-$) would return $X_A = e_A(\chi \cup \alpha)$ and $X_B = e_B(\chi \cup \alpha)$. Note that $X_A > X_B > \mu^*$, so Algorithm~\ref{alg:two} would return $\mathbf{\Omega}$ when run with inputs $X_A$, $X_B$, and $\mu^*$. The same analysis holds for $P^-$ with $\mu$ and $p$ incremented by $\frac{\epsilon}{3}$ and $\frac{\delta}{3}$, respectively, so applying Lemma~\ref{lemma:cases} implies that the case $\mathbf{\Omega}$ is true.

On the other hand, only $A$ would be a candidate state under the prior $P$. Algorithm~\ref{alg:one} run with the prior $P$ would return $X_A = e_A(C_C\cup \alpha)$ and $X_B = e_B(C_C \cup \alpha)$. Run with these inputs (and $\mu^*$), Algorithm~\ref{alg:two} would return $\mathbf{A}$. The same analysis holds for $P$ with $\mu$ and $p$ incremented \textit{and} decremented by $\frac{\epsilon}{3}$ and $\frac{\delta}{3}$, which by Lemma~\ref{lemma:cases} implies that the case $\mathbf{A}$ is true. We can conclude that these values of $\mu$ and $\mu^*$ must be excluded by our promise for any input to \textsc{Promise Revolt} for which our assumed constraints hold.

An argument similar to this counterfactual argument suffices to exclude any value of $p$ or $\mu$ that is insufficiently  far from a crucial decision threshold in Algorithm~\ref{alg:one} (along with associated constraints on $\mu^*$) present in an instance of \textsc{Promise Revolt}.
\end{proof}

Now we can discuss the information that we gain from solving \textsc{Promise Revolt}. In doing so, it is useful to refer to Figure~\ref{fig:equilibria}, which contains the following illustration: Given $\mu^*$, we identify all values of $q$ for which a revolt of size $\mu^*$ is supported with probability at least $q$ (dark blue and yellow regions), all values of $q$ for which revolt of size $\mu^*$ may be supported with probability $q$ (light blue, yellow, and grey regions), and all values of $q$ for which a revolt of size $\mu^*$ is supported with probability strictly less than $q$ (white regions). The blurry regions between distinct colors represent the inputs to \textsc{Promise Revolt} for which two of the cases would overlap (and are therefore excluded from the set of inputs by the ``promise'').

If we could perfectly solve \textsc{Revolt}, then we would be able to perfectly define the boundaries---there would be no blurry regions between distinct colors, as in Figure~\ref{fig:equilibria}. The shape of the resulting figure would characterize the equilibria of the network revolt game in the following sense: The blue column (corresponding to values of $\mu^*$ for which only case $\mathbf{\Omega}$ of \textsc{Promise Revolt} is true) would show the sizes of revolts that are supported in equilibrium with high probability regardless of the state, the next yellow and white column (corresponding to values of $\mu^*$ for which only case $\mathbf{A}$ is true) would show the sizes of revolts that are supported in equilibrium with high probability given that the state is $A$, and the last grey and white column (corresponding to values of $\mu^*$ for which case $\bm{\emptyset}$ is true) would show the sizes of revolts that, with high probability, are not supported in equilibrium (or, equivalently, the sizes of revolts that are supported in equilibrium with low probability). 

Although we cannot perfectly solve \textsc{Revolt}, as shown in Figure~\ref{fig:equilibria}, solving \textsc{Promise Revolt} with Algorithm~\ref{alg:three} allows us to approximate this shape. By choosing $\epsilon$ to be very small, we can make the boundary cases only relevant for very small subsets of possible values of $\mu^*$ and values of $\mu$ in the prior $P$. And, as we have shown via the proof of Theorem~\ref{theorem:alg}, as $n$ grows, $\delta(n)$ quickly becomes very small. Consequently, the range of possibly forbidden values of $p$ in the prior $P$ also quickly becomes small and the probabilities for each class of distinct equilibria described above (when they exist), converge to $1$, $\Pr[\text{State} = A]$, and $0$, respectively. 

\begin{figure}
    \centering
    \includegraphics[width=0.63\textwidth]{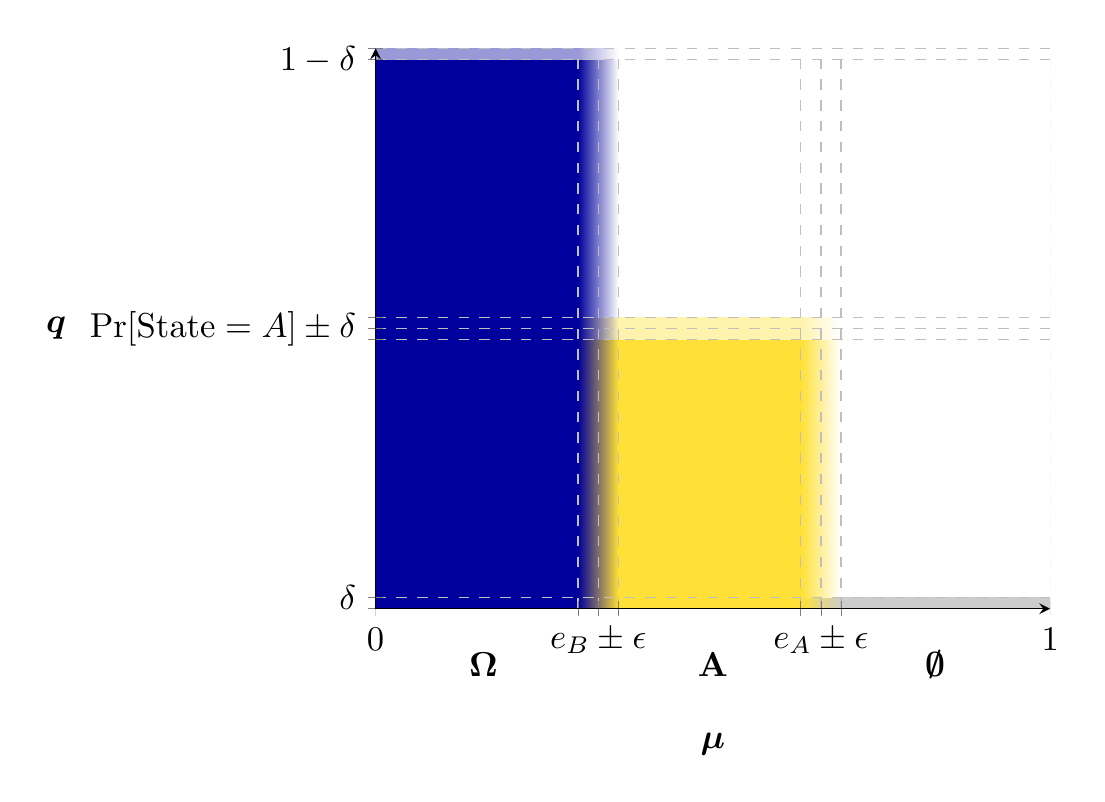}
    \caption{Illustrating what we learn from solving \textsc{Promise Revolt} in an archetypal case.\\ Here, $e_B$ and $e_A$ denote the expected size of the revolt supported in states $B$ and $A$, respectively.}
    \label{fig:equilibria}
\end{figure}

Lastly, we note that an interesting and surprising corollary to our analysis above is that Algorithms~\ref{alg:one}, \ref{alg:two}, and \ref{alg:three} never use any information about the graph beyond the \textit{degree sequence} of $G$---an anonymous list that records the degree of each vertex in the graph. That is, our algorithms require a list of the degrees of the agents, but do not require that any degree value be labelled with the identity of any agent, nor do they require any further information about the set of edges in the graph. Rather, the results of the algorithms are valid for any graph that is consistent with the provided degree sequence, because the concentration of the fraction of candidate agents and agents of each type supersedes the additional structure imposed by any concrete edge set consistent with the degree sequence.

We record this fact with the following proposition:

\begin{proposition}
To compute \textsc{Promise Revolt} in polynomial time, we only require the degree sequence of the graph $G$; the graph itself is not required.
\end{proposition}

\section{Broadening the Setting}
\label{section:broadening}
The core intuition from our analysis above actually applies to a broader class of settings than that of Section~\ref{subsection:fundamental-case}. The additional complexity present in the broader settings generally warrants additional complexity in the analysis, but the core of the argument follows the same reasoning in each case. As a result, we briefly consider the broader settings in the following sections and illustrate how the logic of the initial analysis can be extended to cover those cases. In doing so, we focus on changes to Algorithm~\ref{alg:one}; the subsequent changes required to adapt Algorithms~\ref{alg:two} and \ref{alg:three} are straightforward.

\subsection{Smallest Supported Revolt}
Algorithm~\ref{alg:one} computes the expected size of the largest supported revolt in each state and Algorithm~\ref{alg:two} compares the result with some given value $\mu^*$ to see in which states, if any, the size of supported revolt is at least $\mu^*$. However, this setting has a natural symmetry, which allows us to similarly compute the expected size of the smallest supported revolt in each state with Algorithm~\ref{alg:one} (and compare the results to some given value using Algorithm~\ref{alg:two}). 

To do this, we define $D_T^{A'} = D_T^B$, but swapping the probabilities of $\alpha$- and $\nu$-type agents (so that $\Pr_{D_T^{A'}}[\alpha] = \Pr_{D_T^B}[\nu]$, and vice versa) and similarly define $D_T^{B'} = D_T^A$, swapping the probabilities of $\alpha$- and $\nu$-type agents. These definitions retain the required property that $\Pr[\nu | \text{State} = A] < \Pr[\nu | \text{State} = B]$. Then, we run Algorithm~\ref{alg:one} with the input $G$ and the prior $P = (1 - p, 1 - \mu, D_T^{A'}, D_T^{B'}, D_S)$.

Intuitively, this corresponds to the smallest supported revolt for the following reason: An agent of type $\chi$ feels secure enough to revolt when she believes with probability at least $p$ that at least a $\mu$ fraction of agents will feel secure enough to revolt. Conversely, she does not feel secure enough to revolt when she believes with probability at least $1 - p$ that at least a $1 - \mu$ fraction of agents will not feel secure enough to revolt. As a result, in the setting with the modified inputs, as described, agents who feel secure enough to revolt (including $\alpha$-type agents) correspond precisely with agents who do not feel secure enough to revolt in the original setting (just as $\alpha$-type agents in the modified setting correspond to $\nu$-type agents in the original setting). Algorithm~\ref{alg:one} computes the expected size of the largest supported revolt in each state, so in the modified setting, it computes a value that corresponds to the expected size of the largest group of agents, in some equilibrium of the original setting, who do not feel secure enough to revolt. Since this value is maximized, the size of the supported revolt is minimized. The precise expected size of that minimum revolt in each state is $1 - X_s$ for each $X_s$ returned by the algorithm.

\subsection{General Graphs}
Suppose that we impose no restriction on the degree of the vertices in $G$ in the statement of Theorem~\ref{theorem:alg}. It turns out that our results still hold. The existence of high-degree $\chi$-type vertices (i.e. those that would not exist when we assume an $O\left(\sqrt[3]{n}\right)$ bound on the maximum degree of any vertex), complicates the analysis, but only in the case where $A$ is the only candidate state. When $A$ and $B$ are both candidate states, we only need to calculate the expected number of all $\chi$-type agents, regardless of their degree. When there are no candidate states, $\chi$-type vertices are irrelevant.

When $A$ is the only candidate state, however, there is a key difference:  The presence of high-degree agent contexts in the set of candidate contexts would affect our earlier analysis using Theorem~\ref{theorem:dependent_bound} (in the proof of Lemma~\ref{lemma:concentration-candidate}), because the contexts of high-degree agents are correlated with many other agents' contexts.

Because of this, though, high-degree agents (with degree at least $c \cdot \sqrt[3]{n}$ for sufficiently large $c$)  have a unique perspective on the graph; their contexts contain a significant amount information that they can use in determining the state. More concretely, suppose that agent $i$ is a high-degree agent and let $N(i)$ be the set of neighbors of $i$ in $G$. 

Let $X = \sum_{j \in N(i)} X_j$, where 
\[
    X_j = \begin{cases}
            0, & \text{if $j$ is type $\nu$},\\
            1, & \text{if $j$ is type $\alpha$ or type $\chi$}.
          \end{cases}
\]
Now, for any $\epsilon_0 > 0$, applying a standard Chernoff-Hoeffding bound, we have:
\begin{align}
    \Pr\left[|X - \E[X]| \geq \epsilon_0 \left (c \cdot \sqrt[3]{n}
    \right)\right] \leq 2 \exp \left(\frac{-2{(\epsilon_0 c)}^2 n^{\frac{2}{3}}}{n^{\frac{1}{3}}}\right) \in \frac{1}{\exp \left( \Omega_{\epsilon, P} \left(\sqrt[3]{n} \right) \right)}.
\end{align}

In particular, if we choose $\epsilon_0$ such that $| \E[X | \text{State} = A] - \E[X | \text{State} = B]| > 2 \epsilon_0 \left (c \cdot \sqrt[3]{n} \right)$, then agent $i$ (and by the same argument, any high-degree agent) can correctly determine the state with enough accuracy that their error can be absorbed into the $\delta$ error term with an appropriate choice of $\delta(n)$. As a result, when the state is $A$, all high-degree $\chi$-type agents will behave like candidate agents (recall that $A$ is a candidate state).  

However, for computational purposes, we will not lump them in with the low-degree candidate agents, because Lemma~\ref{lemma:concentration-candidate} can only apply to low-degree candidates.  Instead, in the first \textbf{if} statement after the condition that $A$ is the only candidate state (the line in Algorithm~\ref{alg:one} that reads ``\textbf{if} $e_A(C_C \cup \alpha) \geq \mu$ \textbf{then}''), we calculate $e_A(C_C \cup \alpha \cup H_{\chi})$ instead of just $e_A(C_C \cup \alpha)$, where $H_{\chi}$ refers to the set of high-degree $\chi$-type vertices. If this is at least $\mu$, then when we calculate $X_A$, we include $H_{\chi}$. However, when we calculate $X_B$, we only include $H_{\chi}$ if $e_B(C_C \cup \alpha \cup H_{\chi}) \geq \mu$, since those high-degree agents will not $p$-believe that the state is $A$ regardless of the state the way that (low-degree) candidate agents will. If $e_B(C_C \cup \alpha \cup H_{\chi}) < \mu$, we calculate $X_B = e_B(C_C \cup \alpha)$, as described in Algorithm~\ref{alg:one}.

Next, we must show it is possible for the agents (and the algorithm) to accurately calculate the expected number of high-degree agents of type $\chi$ and know that the actual number of such agents sufficiently concentrates around that expected number. Here, we can rely on the fact that the agents and the algorithm know the degree sequence of the graph. There is some subtlety involved: If the number of high-degree agents is small---less than $\epsilon n$---then we cannot provide a very useful concentration bound for the number of high-degree agents of type $\chi$. However, we can essentially ignore the high-degree agents in this case and absorb the error we incur by ignoring them into our choice of $\delta(n)$.

On the other hand, if there are at least $\epsilon n$ high-degree agents, we can again use a standard Chernoff-Hoeffding bound to conclude that the actual number of high-degree $\chi$-type agents will, with high probability, be close to the expected number of such agents. The error here will be smaller than the error from Lemma~\ref{lemma:concentration-candidate} and so can be absorbed there. 

Finally, there is one additional subtlety we must address for high-degree agents. While it is true that, as previously mentioned, their contexts contain a significant amount of information that they can use in determining the state, their contexts actually contain more than just information about the state---they contain information about the actual realization of types for a significant number of agents. This point is the primary conceptual reason to make the distinction between high-degree and low-degree agents in the first place. Consequently, we need to show that high-degree agents tend to behave as if they only knew the state, when in fact it is possible that the number of agents of a certain type in their context differs greatly from its expectation and as a result they have additional information to use beyond the state of the world. For this, we again appeal to our previous argument that resulted in the bound expressed by the inequality (1), above.

That argument demonstrates that it is highly improbable for the information in a high-degree agent's context to contradict what their belief would be solely given knowledge of the state. For example, if a high-degree agent uses her context to determine that the state is $A$, with high probability it will not also be the case that the context that she uses to make that determination has (far) fewer agents of any type than would be expected given that the state is $A$. Consequently, a high-degree agent will tend to act as if she is just calculating expectations and acting off of them (like a low-degree agent would), even though in actuality she has quite a bit of additional information.

\subsection{A Larger Set of States}
Suppose that there are $m$ states in $S$, with an associated probability distribution $D_T^s$ over the set of three agent types $T$ for each $s \in S$.  
Intuitively, the key insight is similar to the case when $A$ and $B$ are both candidate states in our fundamental setting: When an agent $p$-believes that the true state is in some subset of the states, then she must believe that her $\mu$ threshold would be satisfied in each of the states in that subset in order for her to feel secure enough to revolt. 

Guided by this intuition, we modify Algorithm~\ref{alg:one} in the following way: As before, we identify the set of candidate states by determining for which states $e_s(\chi \cup \alpha) \geq \mu$. If all the states are candidate states, then as in the analogous instance for the two-state setting, there is no need to reason about belief. So, we set $X_s = e_s(\chi \cup \alpha)$ for each state $s$. On the other hand, if not all of the states are candidate states, then we do need to reason about belief. First, we set $C_C$, the set of candidate contexts, to contain the contexts of all agents who believe with probability at least $p$ that the state is in some subset of the candidate states. Then, for each state we compute $e_s(C_C \cup \alpha)$. If for some state $s'$ in the set of candidate states $e_{s'}(C_C \cup \alpha) < \mu$, then we remove $s'$ from the set of candidate states and recompute the set of candidate contexts. We repeat this iterative removal procedure until $e_s(C_C \cup \alpha) \geq \mu$ for all states $s$ in the current set of candidate states. Once this is true (or the set of current candidate states is empty), then $X_s = e_s(C_C \cup \alpha)$ for all the current candidate states and $X_s = e_s(\alpha)$ in the remaining states.

Essentially, this procedure succeeds, because each time we define the set of candidate states, we are as optimistic as possible about the number of agents who may feel secure enough to revolt. This, though, is not too optimistic---i.e. we do not include any superfluous agents---because when all of the states in the set of candidate states satisfy $e_s(C_C \cup \alpha) \geq \mu$, then all of the agents with candidate contexts feel secure enough to revolt, based on our intuition above.

More formally, we can use an inductive argument: The initial set of candidate contexts contains the contexts of every agent for whom it is possible that they will feel secure enough to revolt. At each step when we remove contexts from the set of candidate contexts, any agent with a context that is removed will never feel secure enough to revolt. They do not feel secure enough to revolt with the given set of candidate contexts and $\alpha$-type agents, and this set, at each step in the iteration, is maximal. Thus, at each step the set $C_C \cup \alpha$ contains every agent who feels secure enough to revolt. So, when the iteration terminates, the resulting (possibly empty) set of agents in $C_C \cup \alpha$ that all feel secure enough to revolt must be the largest such set.

The number of steps in this procedure is polynomial in $m$, the number of states, so we consider it still to be efficient.

It is worth noting that this result is driven by the specific way we have defined the problem. In our setting, we are only required to consider the equilibrium with the maximum revolt size in each state. As we have shown, enumeration of all of the subsets of states is not required to compute this. However, if we were to try to characterize all of the equilibria of our network revolt game, without making any further assumptions on the states, a stronger method, potentially including enumeration of all of the subsets of states, would be necessary.

\section{Experiments}
\label{section:experiments}
Beyond the theoretical insight they provide, our algorithms from Section~\ref{section:applying} also have practical utility for exploring the relationship between strategic coordination (i.e. a revolt supported in some equilibrium) and various parameters---both parameters of network models from which the graph $G$ can be generated and parameters of the prior $P$. To demonstrate this utility, we return to the concrete setting of our motivating example. For convenience, the relevant parameters from this setting are summarized in Table~\ref{tab:motiv_example_params}.

\begin{table}[t]
\caption{Summary of the prior $P$ from the motivating example.}
\begin{subtable}[t]{.48\textwidth}
   \centering
   \caption{Summary of the parameters.}
   {\renewcommand{\arraystretch}{1.25}
\begin{tabular}{ |c|c| } 
 \hline
 Parameter & Value\\
 \hline
  $p$ & $\frac{2}{5}$ \\
  \hline
  $\mu$ & $\frac{1}{2}$\\
  \hline
\end{tabular}
}
    \label{tab:prior}
\end{subtable}
\begin{subtable}[t]{.48\textwidth}
   \centering
   \caption{Summary of the distributions.}
   {\renewcommand{\arraystretch}{1.25}
\begin{tabular}{ |c|c|c| } 
 \hline
  & $s = A$ & $s = B$ \\
  \hline
  $\Pr[\text{State} = s]$ & $\frac{1}{2}$ & $\frac{1}{2}$\\
  \hline
  $\Pr[\alpha | \text{State} = s]$ & 0 & 0\\
  \hline
  $\Pr[\chi | \text{State} = s]$ & $\frac{4}{5}$ & $\frac{1}{5}$\\
  \hline
  $\Pr[\nu | \text{State} = s]$ & $\frac{1}{5}$ & $\frac{4}{5}$\\
 \hline
\end{tabular}
}
   \label{tab:type_distributions}
\end{subtable}

\label{tab:motiv_example_params}
\end{table}

In the following subsections, we focus on how the size of the largest supported revolt in expectation in each state varies with the chosen parameters. That is, we apply Algorithm~\ref{alg:one} to compute our dependent variables.

However in general, Algorithm~\ref{alg:three} also has practical utility. For example, we are able apply Algorithm~\ref{alg:three} to instantiate concrete constructions of Figure~\ref{fig:equilibria} in a given setting with an appropriate choice of $\delta$ (given $\epsilon$ and the prior). For the setting of our motivating example, this is illustrated in Figure~\ref{fig:example}.

\begin{figure}
    \centering
    \includegraphics[width=0.63\textwidth]{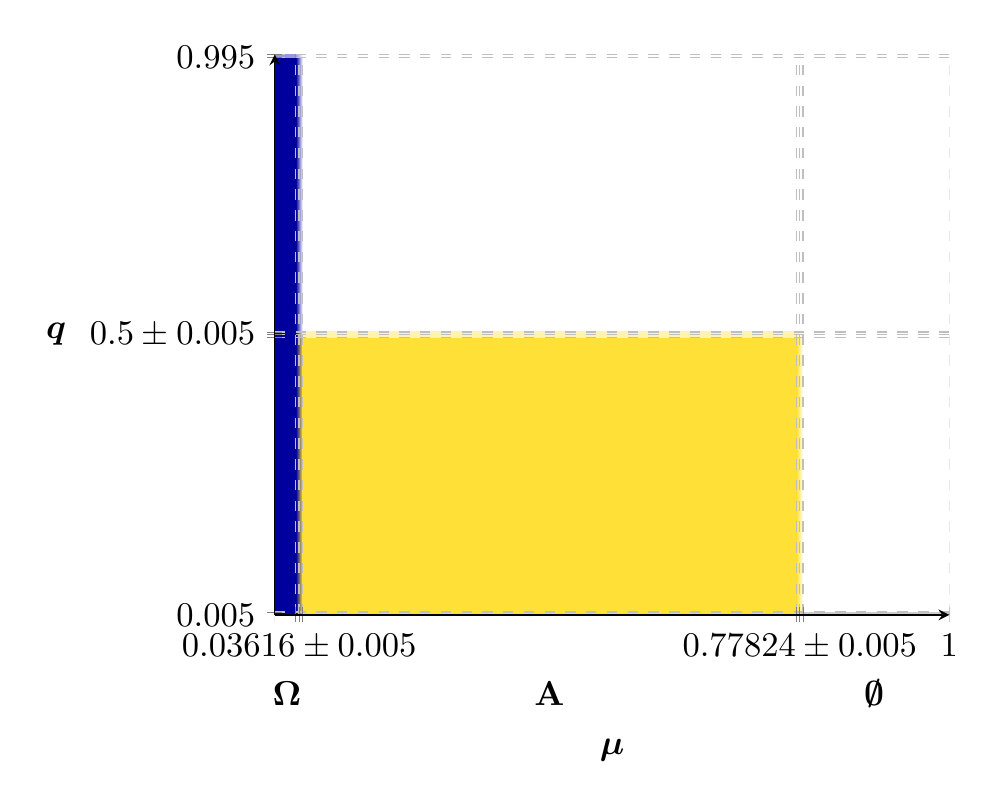}
    \caption{Instantiating Figure~\ref{fig:equilibria} in the setting of our motivating example, with $\delta = \epsilon = 0.005$.}
    \label{fig:example}
\end{figure}

\subsection{Varying Parameters of Network Models}
\label{subsection:varying_network_params}
In this experiment, we explore the relationship between the size of the largest supported revolt in expectation---denoted as in Figure~\ref{fig:equilibria} by $e_s$ for the state $s \in \set{A, B}$---and the parameters of the various network models described below. $d$-regular graphs, which have constant degree sequences, serve as a baseline which helps us interpret the results from the other three network models. Each of the other models are each common social network models that generate graphs according to two parameters.
The first parameter for each social network model is the number of vertices $n$, which we set to $n = 1000$ for each model.
The second parameter is unique to the model (we discuss the details for each model in the associated paragraph below). We treat these parameters as the independent variables. 

\begin{figure}
\centering
\begin{subfigure}{.5\textwidth}
  \centering
  \includegraphics[width=\linewidth]{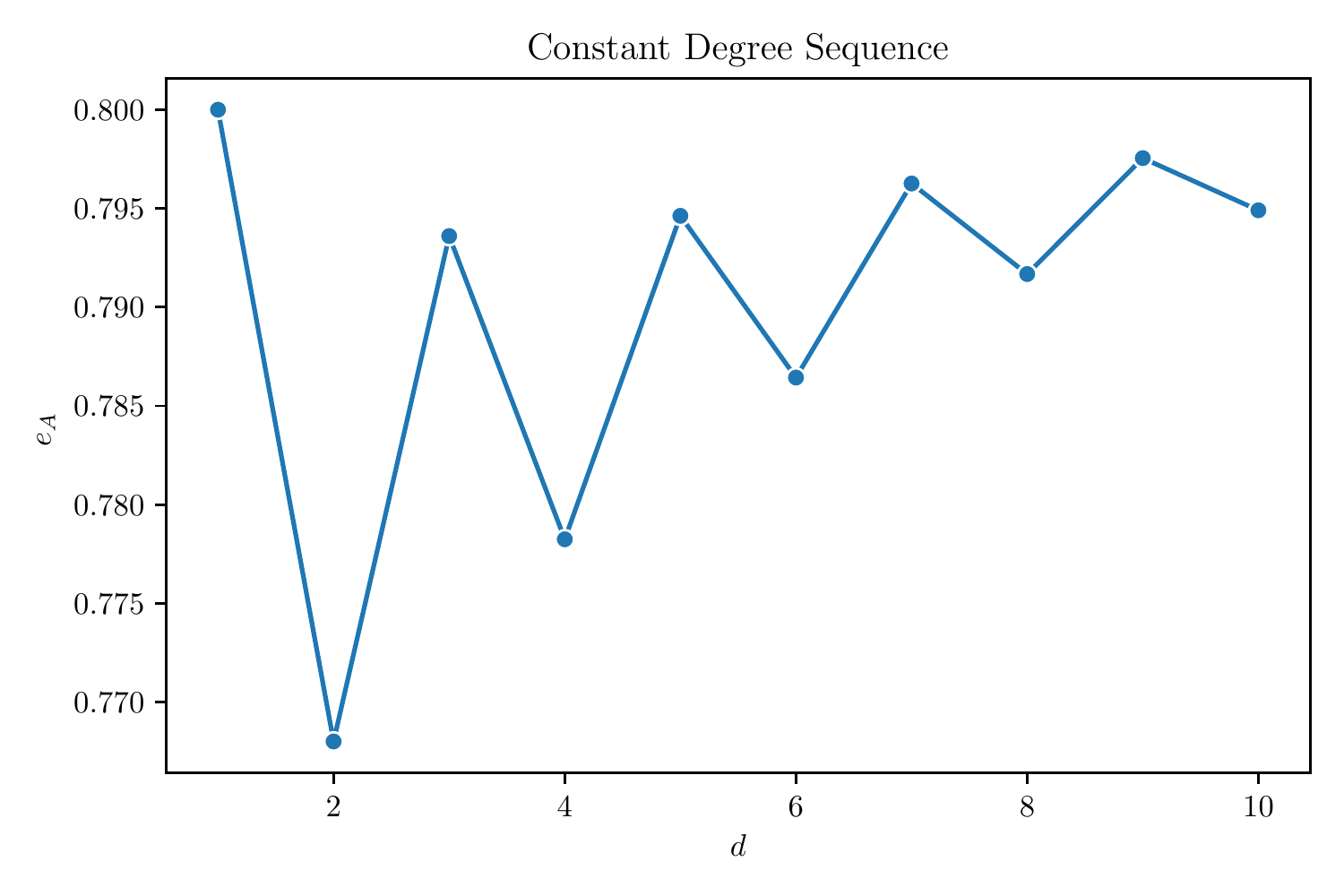}
\end{subfigure}%
\begin{subfigure}{.5\textwidth}
  \centering
  \includegraphics[width=\linewidth]{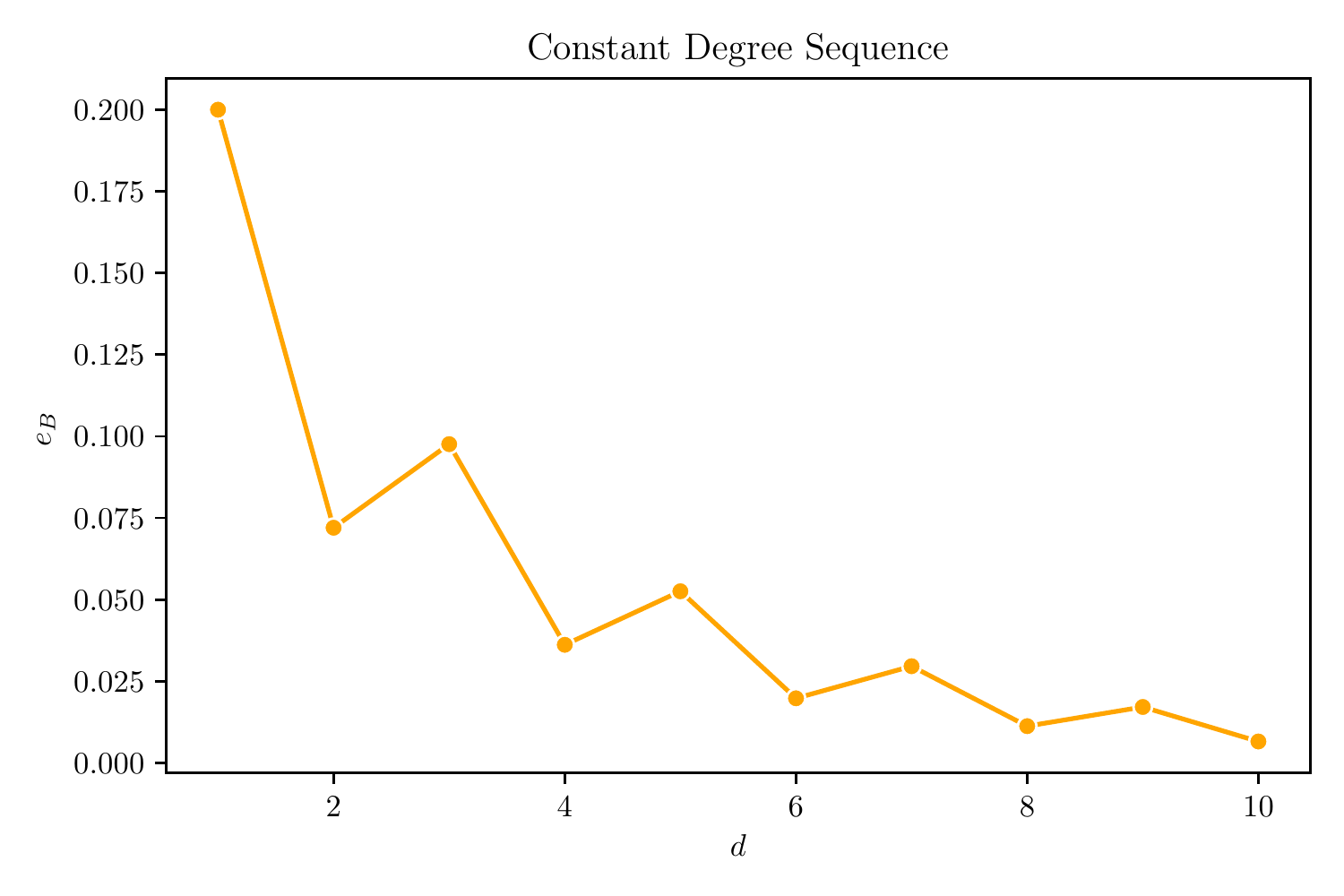}
\end{subfigure}

\begin{subfigure}{.5\textwidth}
  \centering
  \includegraphics[width=\linewidth]{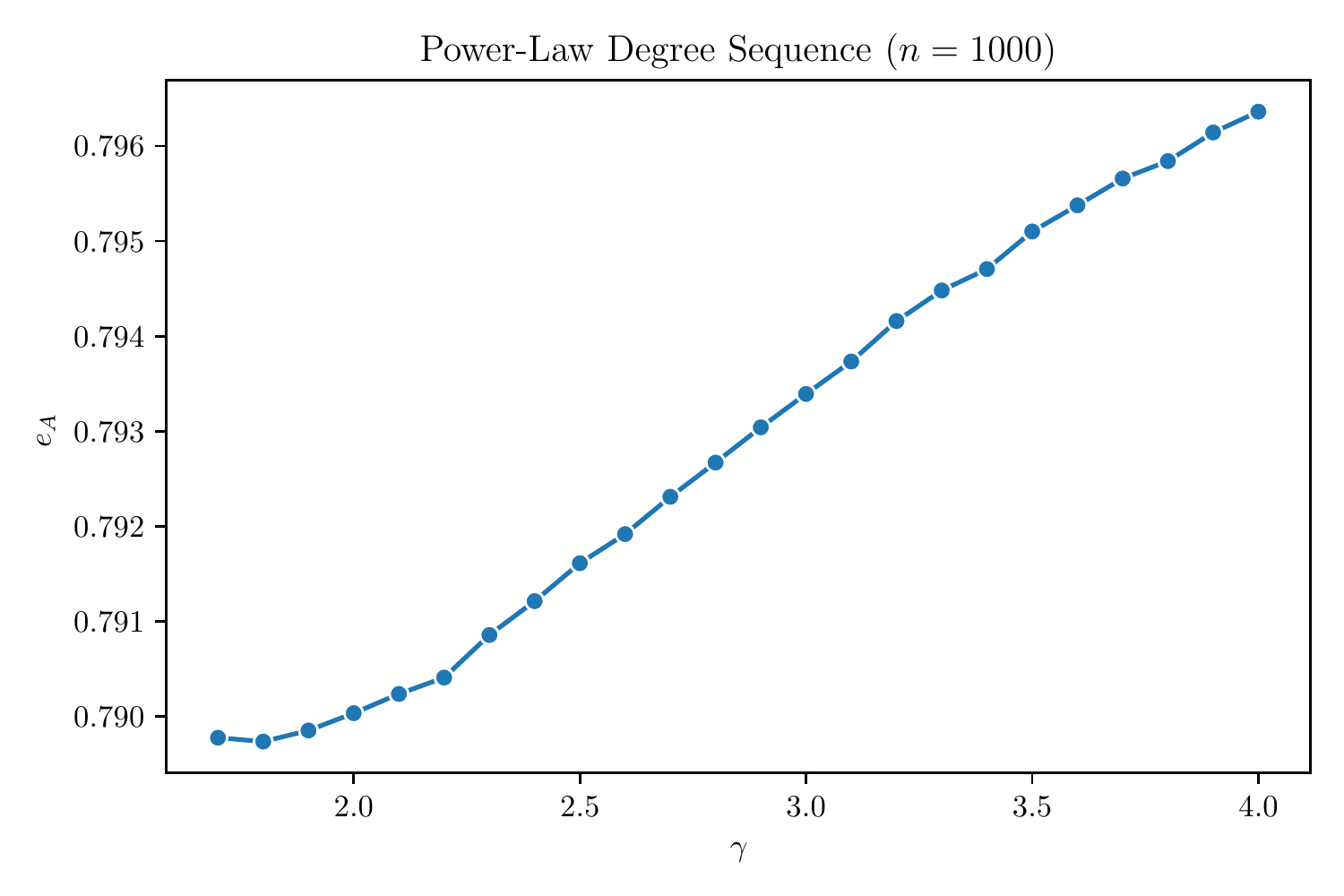}
\end{subfigure}%
\begin{subfigure}{.5\textwidth}
  \centering
  \includegraphics[width=\linewidth]{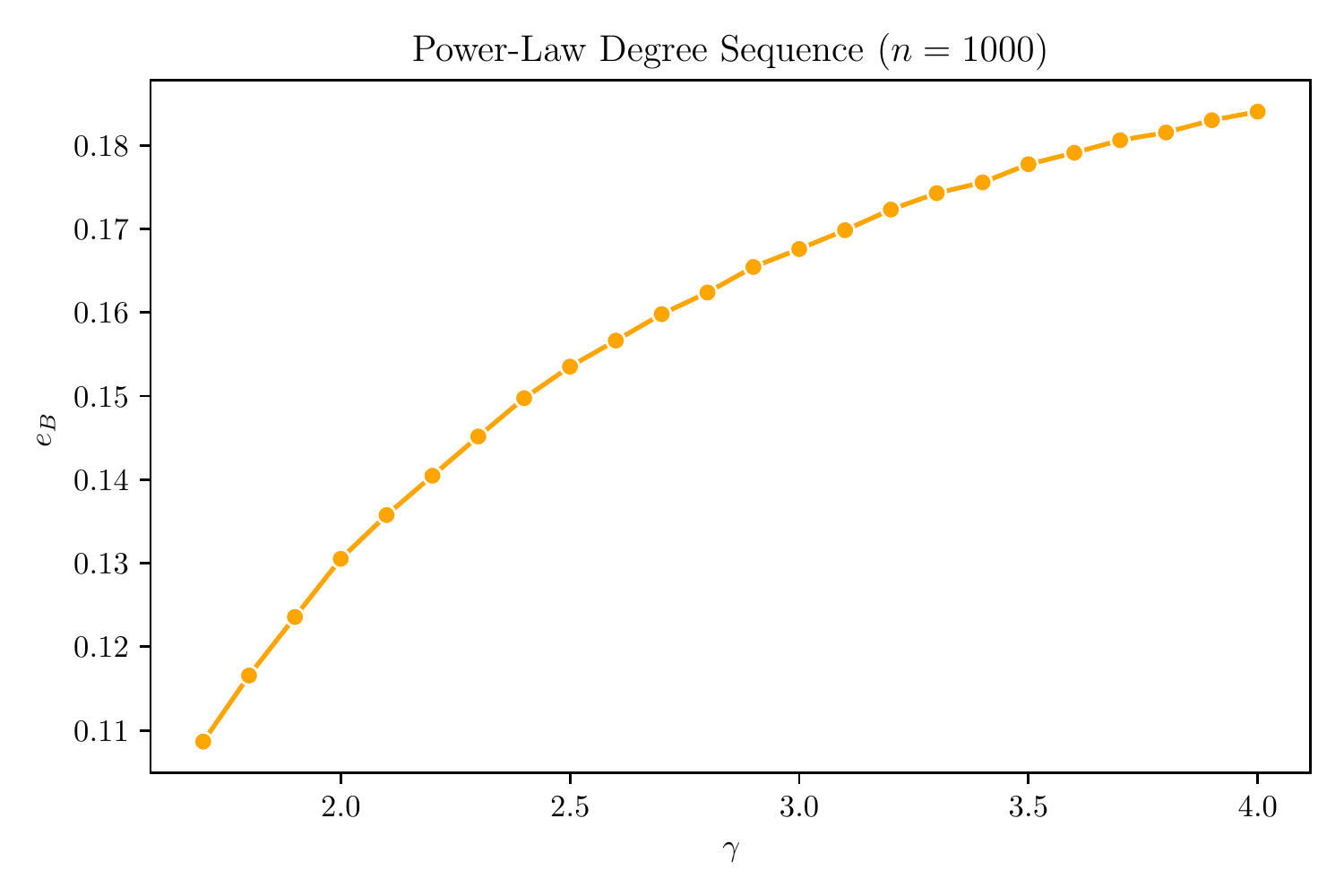}
\end{subfigure}

\begin{subfigure}{.5\textwidth}
  \centering
  \includegraphics[width=\linewidth]{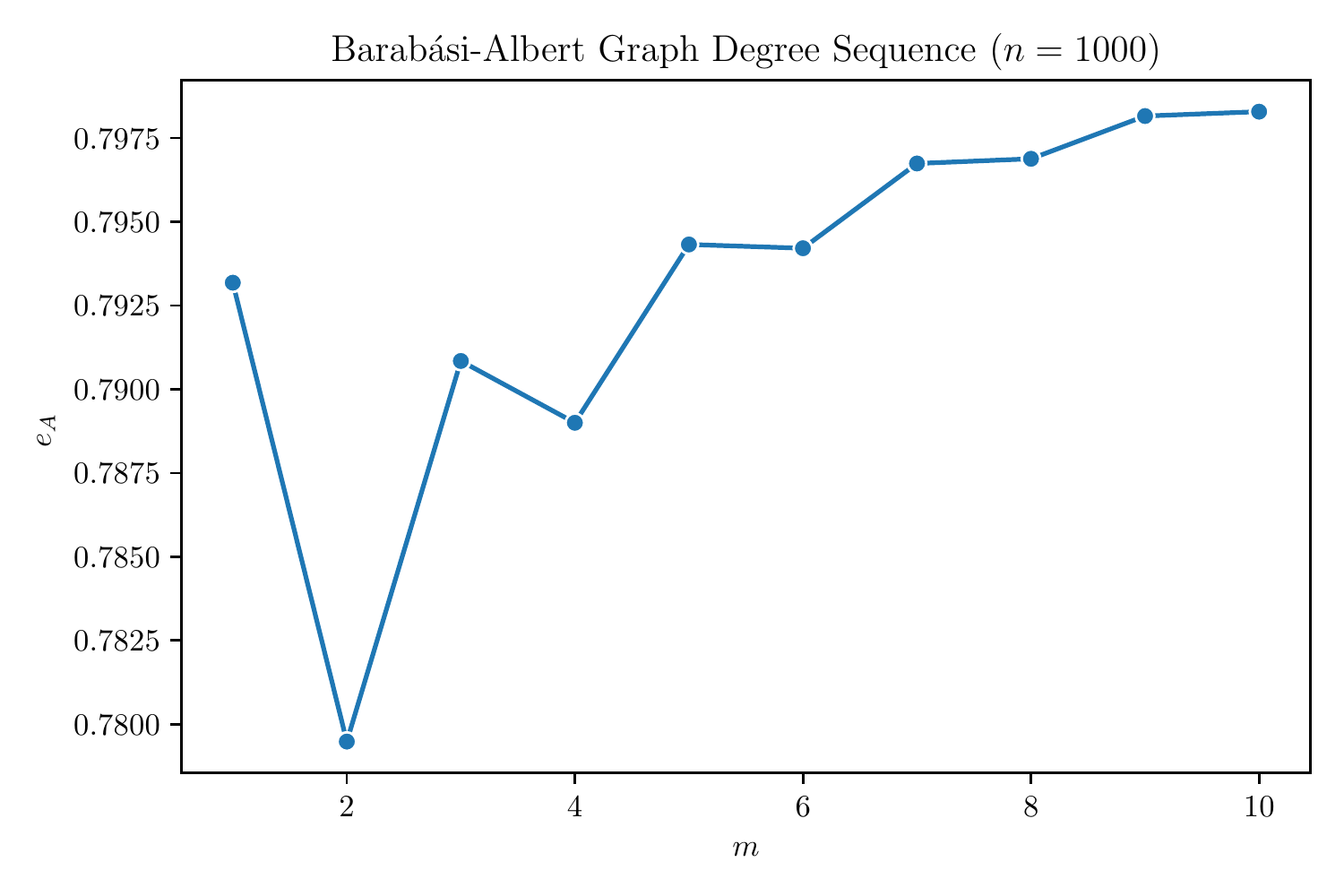}
\end{subfigure}%
\begin{subfigure}{.5\textwidth}
  \centering
  \includegraphics[width=\linewidth]{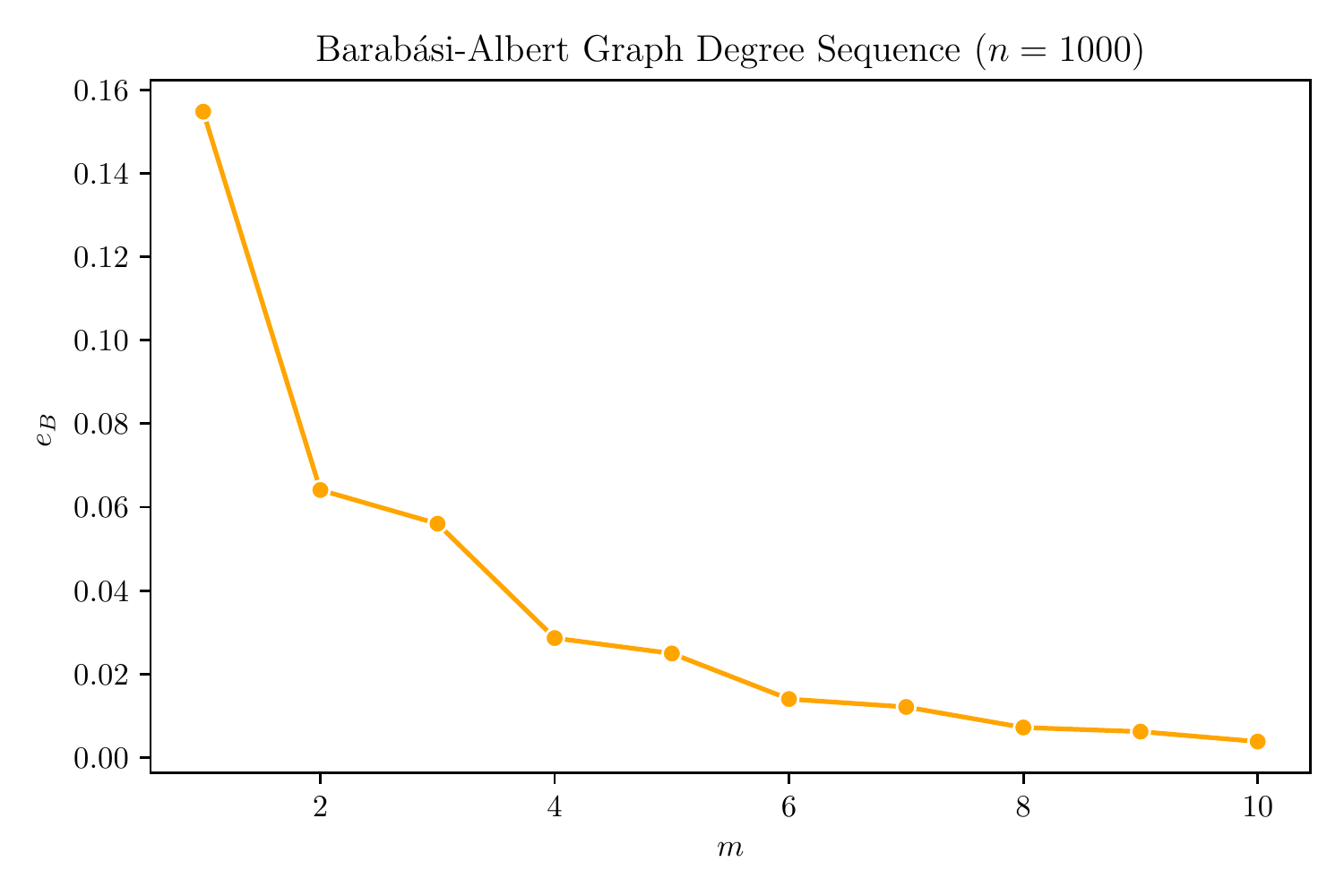}
\end{subfigure}

\begin{subfigure}{.5\textwidth}
  \centering
  \includegraphics[width=\linewidth]{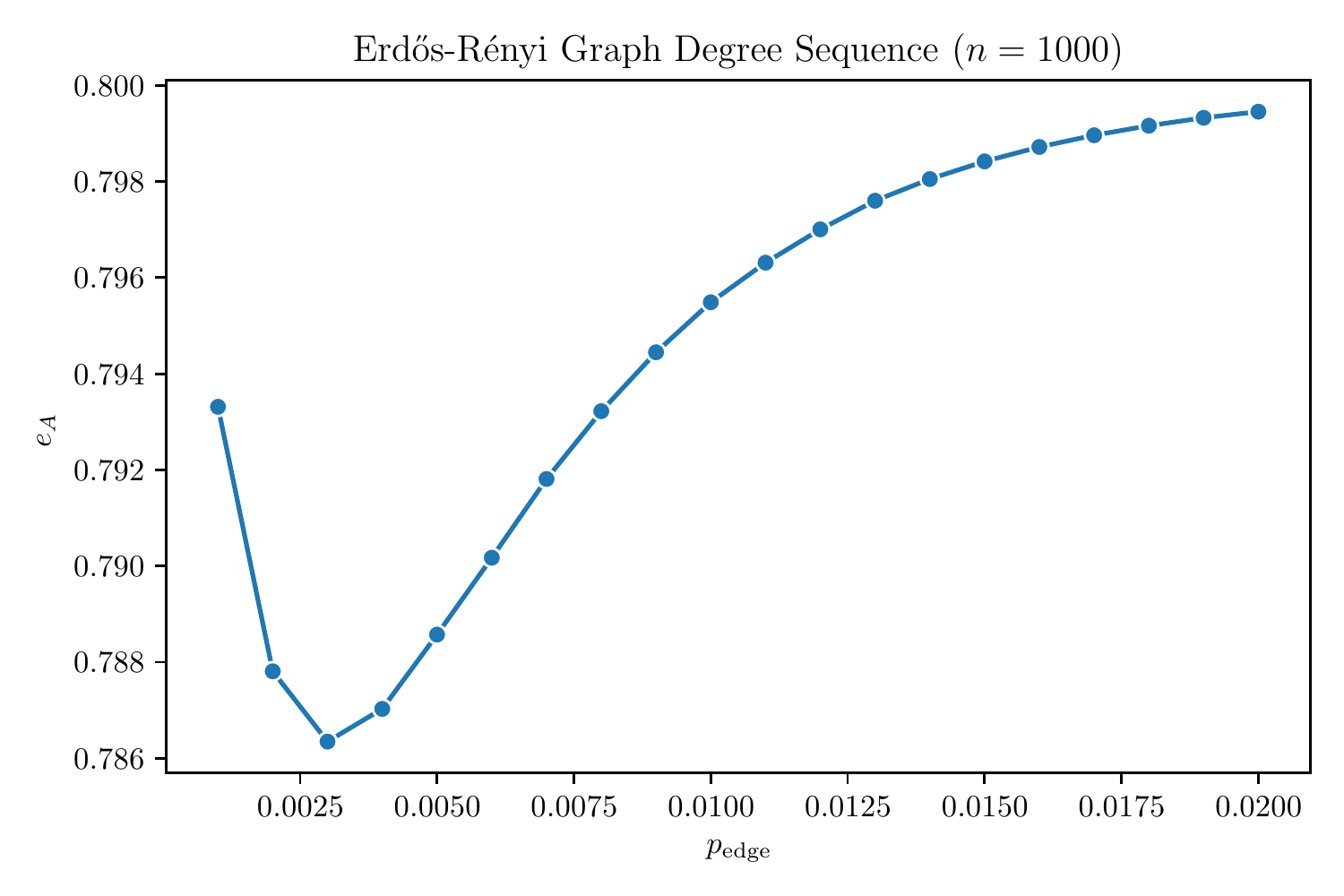}
\end{subfigure}%
\begin{subfigure}{.5\textwidth}
  \centering
  \includegraphics[width=\linewidth]{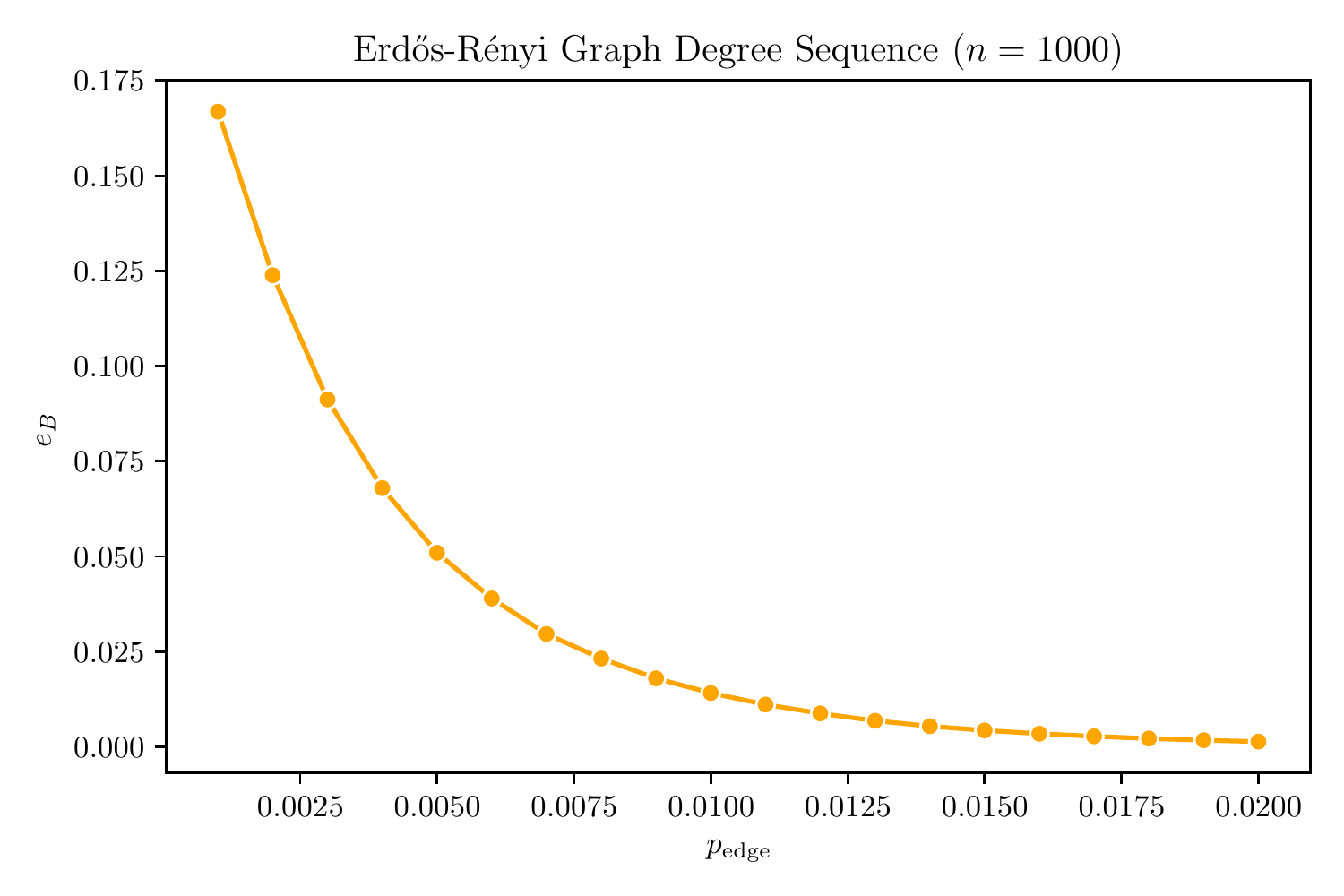}
\end{subfigure}
\caption{Results of Experiment~\ref{subsection:varying_network_params}---Exploring the relationship between $e_A$ and $e_B$ and the underlying parameters of various network models in the setting of our motivating example.}
\label{fig:varying_network_params}
\end{figure}

\subsubsection{Methods}
\label{subsub:network_params_methods}
For constant-degree graphs, the degree sequence is deterministic given the degree $d$. So, we ran an implementation of Algorithm~\ref{alg:one} a single time with a degree sequence consisting of $1000$ entries equal to $d$ and the prior $P$ from our motivating example (see Tables \ref{tab:prior} and \ref{tab:type_distributions}) as input and recorded the output values.

For the social network models, the generated degree sequences are not deterministic. So, for each social network model and for each value of the relevant parameter, we generated 100 graphs of size $n = 1000$ from the model with that particular parameter value. Then, for each of the 100 generated graphs, we ran our implementation of Algorithm~\ref{alg:one} with the degree sequence of the graph and the prior $P$ from our motivating example as input and recorded the mean of the output values (i.e. the average $e_s$ for each state $s$). 

The results are illustrated in Figure~\ref{fig:varying_network_params}.

\paragraph{Constant Degree Sequence.} For constant degree sequences, the relevant parameter is $d$, the degree of each vertex. 

As we noted above, the degree sequence is deterministic given $d$ and it is trivial to generate. Furthermore, the output of Algorithm~\ref{alg:one} is agnostic to $n$, so the choice of $n = 1000$ is somewhat irrelevant. The result would be the same for any $n$, so the key constraint is whether a $d$-regular graph on $n$ vertices exists for given $n$ and $d$. This is the case for any $n \geq d+1$ such that $nd$ is even, and therefore is true for $n = 1000$.

\paragraph{Power-Law Degree Sequence.} For power-law degree sequences, the relevant parameter is $\gamma$, which represents the exponent of a power-law distribution. Specifically, in a graph with a power-law degree distribution, the probability of observing a vertex with degree $d$ is proportional to $d^{-\gamma}$. 

For our experiment, we used the \texttt{powerlaw} package in Python \cite{Alstott2013} to generate integer sequences of length $n = 1000$ which are distributed according to a power-law distribution with exponent $\gamma$ and then used the Havel-Hakimi algorithm implementation from the \texttt{networkx} package in Python to determine whether the sequences corresponded to feasible graph degree sequences \cite{Hagberg2008}. This two-step process was repeated until we had 100 such feasible degree sequences for each value of $\gamma$.\footnote{The range of $\gamma$ in our experiments was somewhat constrained by this feasibility checking. For $\gamma < 1.7$, generating even one sequence that corresponded to a feasible graph was intractable.} 

\paragraph{Barab{\'a}si-Albert Graph Degree Sequence.} For degree sequences from Barab{\'a}si-Albert graphs (also referred to as preferential attachment graphs) \cite{Barabasi1999}, the relevant parameter is $m$, which represents the number of edges that ``incoming'' vertices attach (preferentially) to ``existing'' vertices.\footnote{See Chapter 5 of Jackson's book on networks for an overview of preferential attachment models \cite{Jackson2010}.}

For our experiment, we used a generator from the \texttt{networkx} package in Python to generate Barab{\'a}si-Albert Graphs \cite{Hagberg2008}.

\paragraph{Erd\H{o}s-R{\'e}nyi Graph Degree Sequence.} For degree sequences from Erd\H{o}s-R{\'e}nyi random graphs, the relevant parameter is $p_{\text{edge}}$, which represents the probability that an edge between any fixed pair of vertices will exist in the graph after the edges are sampled.

For our experiment, we used a generator from the \texttt{networkx} package in Python to generate Erd\H{o}s-R{\'e}nyi Graphs \cite{Hagberg2008}.

\subsubsection{Results}
\label{subsub:network_params_results}
As illustrated in the two columns of Figure~\ref{fig:varying_network_params}, the results from both states display interesting phenomena and the distinct states have distinct features of interest.

\paragraph{State $A$.} For each type of degree sequence, we see that $e_A$ displays non-monotonicity with respect to the parameter serving as the independent variable. 

The results from the constant degree sequence are useful in interpreting this phenomenon. First, we see that there are parity effects---agents of type $\chi$ with odd degrees are more likely to sufficiently believe that the state is $A$. In the most extreme case, in fact, agents of type $\chi$ with degree one will always sufficiently believe that the state is $A$, because regardless of their neighbor's type, they will believe that the state is $A$ with probability at least $\frac{1}{2} > \frac{2}{5}$ (the value of $p$ in the prior). Second, controlling for the parity effects and the unique case of the degree-one agents, agents with higher degrees tend to be more likely to correctly determine the state. 

This observations can be applied to explain the non-monotonicity we observe in the results for the degree sequences drawn from social network models. For the Barab{\'a}si-Albert and Erd\H{o}s-R{\'e}nyi graphs, as with constant degree sequence graphs, the modal degree value in the graph increases monotonically with the independent variable parameter, so we observe non-monotonicity in the lower values of the relevant parameter as degree 2 vertices become more common than degree 1 vertices. For the power-law degree sequences, the modal degree value is 1 for all relevant values of $\gamma$, but the frequency of degree 1 vertices increases monotonically with $\gamma$. The non-monotonicity in the lower values of the parameter corresponds to a small range in which degree 2 vertices become more common, and the effect of the increasing frequency of degree 1 vertices is not yet large enough to offset the effect of the greater frequency of vertices of degree 2. 

\paragraph{State $B$.} For each type of degree sequence, the shape of the curves for $e_B$ can be explained with roughly the same explanation used for $e_A$ above. The key differences is that agents of type $\chi$ with degree one are no longer an exception to the trend that higher-degree agents tend to be more accurate in determining the state (after controlling for parity effects). In particular, this explains why we see the non-monotonicity of $e_A$ for the degree sequences drawn from social network models, but not of $e_B$.

The more interesting phenomenon is that the range of $e_B$ is significantly larger than the range of $e_A$ for each graph. This phenomenon is relatively simple to explain: it is uncommon for low-degree agents of type $\chi$---and impossible for such agents with degree one---to observe a context that convinces them that the state is more likely to be $B$, since their own type provides evidence against this. Their own evidence is weighted highly, because their total amount of evidence is small.

Even in light of this explanation, though, it is striking that we see $e_B$ range from including almost every agent of type $\chi$ to almost no agents of type $\chi$ (or just over half of the agents of type $\chi$, in the case of the power-law degree sequences). Even excluding parameter values where degree one agents proliferate (which produce the most extreme values of $e_B$), the range of $e_B$ still varies significantly more than the range of $e_A$. Consequently, when the state is $B$, we could expect that a significant fraction of the agents of type $\chi$ would feel secure enough to a revolt in conditions that would be very unlikely to support a revolt in which they would feel comfortable participating. This phenomenon is of particular interest with regard to the formation and maintenance of unpopular social norms like the culture of excessive alcohol consumption in American colleges and universities and the more general phenomenon of \textit{pluralistic ignorance} \cite{Bicchieri2005, Centola2005, Chwe1999}.

\subsection{Varying \textit{p} in the Prior}
\label{subsection:varying_p}
In this experiment, we explore the relationship between $e_A$ and $e_B$ and the value of the parameter $p$ in the prior. Recall that $p$ represents the belief threshold required for agents of type $\chi$.

\begin{figure}
\centering
\begin{subfigure}{.5\textwidth}
  \centering
  \includegraphics[width=\linewidth]{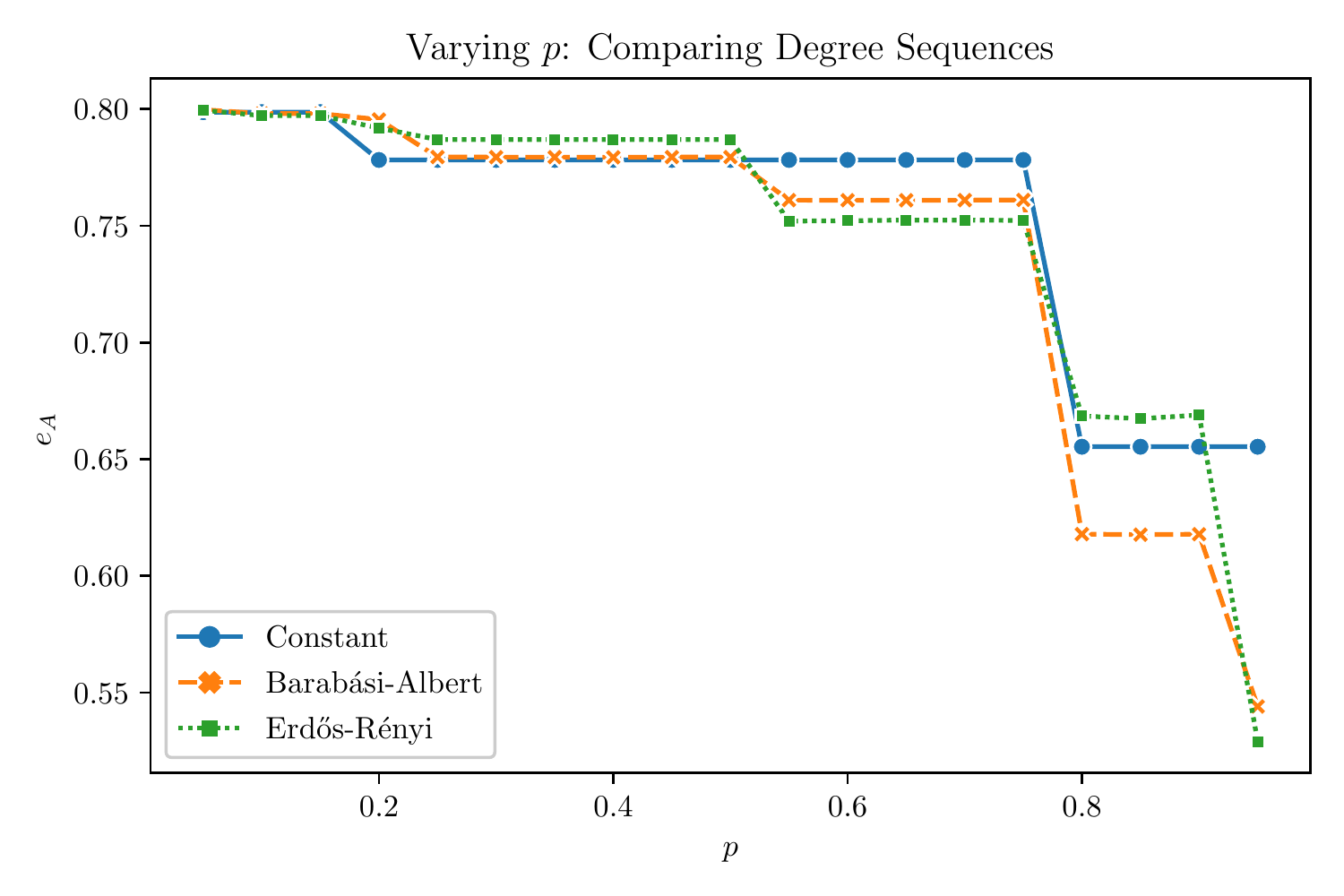}
\end{subfigure}%
\begin{subfigure}{.5\textwidth}
  \centering
  \includegraphics[width=\linewidth]{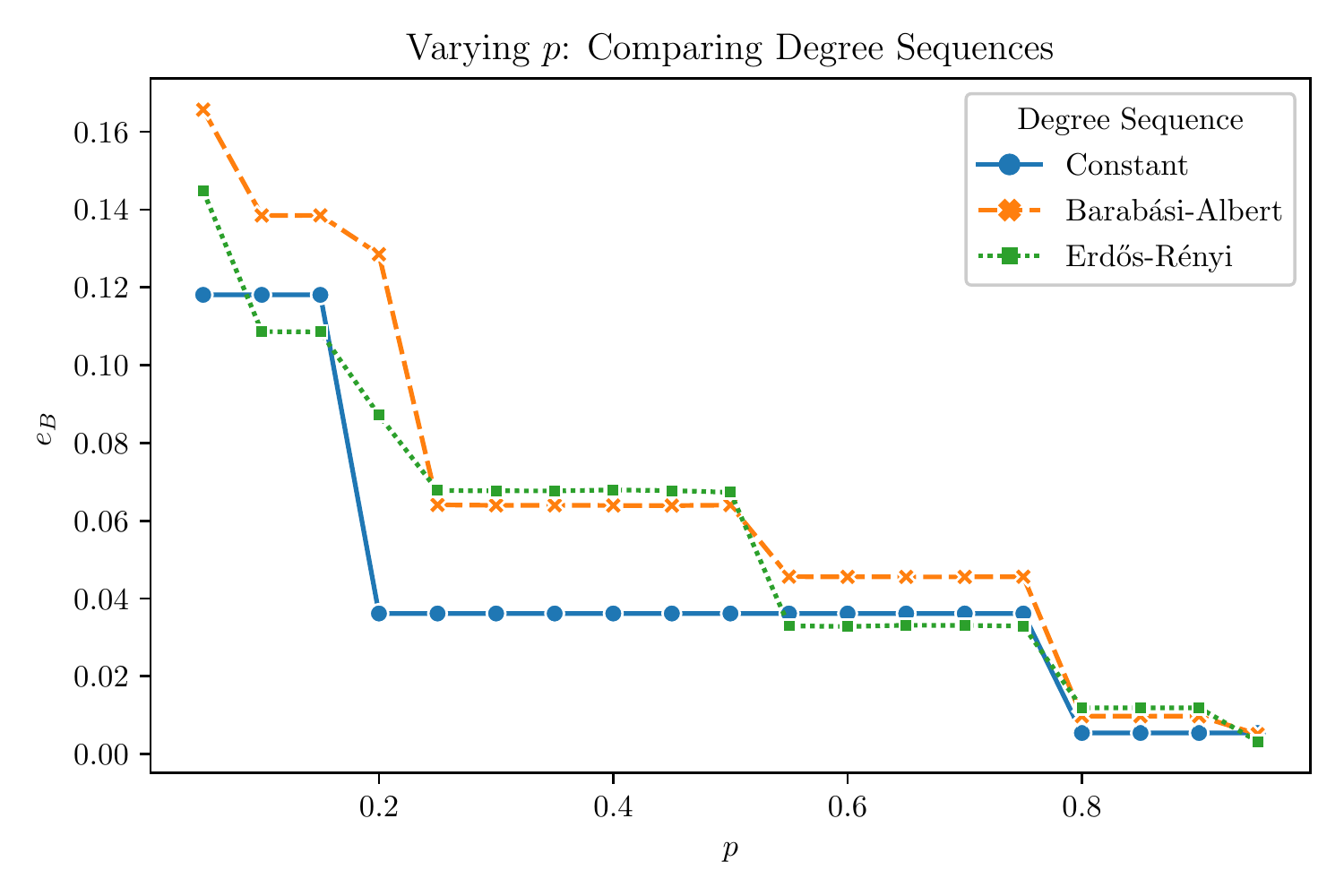}
\end{subfigure}

\caption{Results of Experiment~\ref{subsection:varying_p}---Exploring the relationship between $e_A$ and $e_B$ and the value of the parameter $p$ in the prior for various network models with average degree 4 in the setting of our motivating example.}
\label{fig:varying_p}
\end{figure}

\subsubsection{Methods}
\label{subsub:p_methods}
For this experiment, we first fix the underlying parameters of the degree sequences that we varied in the previous experiment so that for each degree sequence, the average or expected degree for each vertex will be 4. Thus, each model is expected to produce degree sequences that have the same average degree, but with different variances. For the constant degree sequence this is accomplished with $d = 4$, for the degree sequences of Barab{\'a}si-Albert graphs it is with $m = 2$, and for the degree sequences of Erd\H{o}s-R{\'e}nyi graphs it is with $p_{\text{edge}} = \frac{1}{250}$.\footnote{In this experiment, including power-law degree sequences for some fixed $\gamma$ would be somewhat redundant; the degree sequences of Barab{\'a}si-Albert graphs follow a power-law distribution with $\gamma = 3$.} As before, for all of the degree sequences, we set $n = 1000$.

With those parameters, we generated graphs and ran Algorithm~\ref{alg:one} as in the previous experiment---running Algorithm~\ref{alg:one} once for the constant degree sequence and generating 100 graphs and recording the average outputs of Algorithm~\ref{alg:one} for the social network models ---for each value of $p$ (i.e. $p$ between $0.05$ and $0.95$, incremented in steps of size $0.05$).

The results are illustrated in Figure~\ref{fig:varying_p}.

\subsubsection{Results}
\label{subsub:p_results}
The behavior of $e_A$ and $e_B$ for each type of degree sequence in this experiment is easy to predict---they both decrease monotonically in steps as $p$ increases. As a result, in this experiment, it is more interesting is to compare across the different degree sequences. As an example, for degree sequences from Barab{\'a}si-Albert graphs, the values of $e_A$ and $e_B$ are somewhat lower than the analogous values for degree sequences from Erd\H{o}s-R{\'e}nyi graphs for relatively high values of $p$---the values where we might expect $p$ thresholds to be for real agents contemplating a potentially costly action like violating an existing social norm. This might indicate that in real social networks---which in terms of degree sequences, among the three models in this experiment, are best approximated by Barab{\'a}si-Albert graphs---we might expect agents to be more conservative about taking costly actions than they would be in random graphs. On the other hand, the opposite is true for $p$ values that are closer to, but still above, $\frac{1}{2}$, so we might expect agents in real social networks to be relatively less conservative for actions with a moderate cost.

\section{Discussion}
\label{section:discussion}
We proposed that the notion of common belief could be relaxed to a notion of factional belief in order to be more suited to social network settings and gave a natural definition of factional belief, drawing heavily from previous work on common belief. We then applied this definition theoretically and experimentally in a setting inspired by prior work about common knowledge and revolt games on networks to show how this definition moves beyond the limitations of previous work by being applicable in general graphs.

\subsection{Open Questions and Future Work}
\label{subsection:future_work}
The most clear direction for future work is to continue to apply factional belief in new settings and use it as a tool to understand strategic coordination and cooperation on networks. In particular, as mentioned in the introduction, the work of Stephen Morris provides many examples of applying common belief as a tool for gaining insight into a diverse range of settings. We believe that applying factional belief can yield similar results in social network settings, and that this paper represents an initial step in that process.

Additionally, though, there are some more subtle technical open questions regarding our definition of factional belief that are also worth exploring in future work that we state below.
\begin{enumerate}
    \item Considering our definition of factional belief from the perspective of an infinite hierarchy of reasoning, akin to the initial definition for common knowledge that we provided (\ref{def:hierarchy}), with regard to the $\mu$ fraction of agents at each step in the hierarchy, it is not required by our definition that this be the same $\mu$ fraction of agents at each step in the hierarchy. However, it is not clear whether or not this should always be the case. Is it possible to describe an event that is common $(p, \ \mu)$-belief, but for which the infinite hierarchy refers to a different $\mu$ fraction of agents at some step? If so, what are the consequences of this for strategic coordination and cooperation?
    
    \item It is not too difficult to modify the setting described in Section 4 to create a model where the underlying event that supports revolt does not necessarily neatly reduce to a single event that is common $(p, \ \mu)$-belief. For example, if there are multiple types of conditionally-revolting agents with different $p$ and $\mu$ thresholds, the event supporting revolt is more like a $\mu$ fraction of agents believe with sufficient probability that their thresholds are satisfied. This event encompasses common $p$-beliefs among certain agents of the same type regarding their $\mu$ thresholds, but not precisely common $(p, \ \mu)$-beliefs. Is there a more general definition of factional belief that allows for the existence of different thresholds for different types of agents to still be encompassed in a single event that is a factional belief among all of the agents who feel satisfied enough to revolt?  
\end{enumerate}
We believe that answering these questions could yield further insight into factional belief that could inform its application in other settings.

%
%
%
%
%
%

\newpage
\bibliography{relaxing_common_belief}

\newpage
\appendix
\section{Proof of Proposition~\ref{prop:one}}
\label{appendix:proof_prop2.4}
Recall that ($\Omega$, $\Sigma$, $\Pr$) is a probability space and that $I$ denotes a set of agents. For our proof below, we borrow the following properties of the belief operator $B_i^{p}$, from Monderer and Samet \cite{Monderer1989}. The properties are numbered for consistency with their paper. 

\begin{proposition}[Proposition 2 from Monderer and Samet \cite{Monderer1989}]
\hfill

For each $0 \leq p \leq 1$, $i \in I$, and $E, F \in \Sigma$,
\begin{align}
    & B_i^{p}(B_i^{p}(E)) = B_i^{p}(E). \tag{6}\\
    & \text{If $E \sse F$, then $B_i^{p}(E) \sse B_i^{p}(F)$.} \tag{7}\\
    & \text{If $\left(E^n\right)$ is a decreasing sequence of events, then }
    B_i^{p}\left(\cap_n E^n\right) = \cap_n B_i^{p}(E^n). \tag{8}
\end{align}
\end{proposition}

Given the above proposition, the proof closely follows Monderer's and Samet's proof of their Proposition 3:
\begin{proof}[Proof of Proposition~\ref{prop:one}] 
The proof is split into two arguments, one for each of the conclusions of the proposition:
\vspace{1 em}

\begin{enumerate}
    \item First, we show that $\left(F_{\mu}^n\right)_{n = 1}^{\infty}$ is a decreasing sequence:
    
    \vspace{1 ex}
    
    By definition of $F_{\mu}^n$, for each $n \geq 1$, there exists $J \sse I$ with $|J| \geq \mu |I|$ such that  $\forall j \in J$, $F_{\mu}^n \sse B_{j}^p(F_{\mu}^{n-1})$. 
    
    \vspace{1 ex}
    
    As a result, by (6) and (7), 
    \[
        \exists \, J \sse I \text{ with $|J| \geq \mu |I|$ s.t. $\forall j \in J$, } B_j^p(F_{\mu}^n) \sse B_{j}^p(B_{j}^p(F_{\mu}^{n-1})) = B_{j}^p(F_{\mu}^{n-1}).
    \]
    
    So, for all $n \geq 1$, we have 
    \[ 
        F_{\mu}^{n+1} \sse \bigcap_{j \in J} B_j^p(F_{\mu}^n) \sse \bigcap_{j \in J} B_j^p(F_{\mu}^{n-1}) \sse F_{\mu}^{n}.
    \]
    
    \vspace{1 ex}
    
    Now, for all $n \geq 1$, let $j \in J$. Then, we have 
    \[
        E^{p, \mu}(F) \sse F_{\mu}^{n + 1} \sse B_j^p(F_{\mu}^n).
    \]
    Therefore, $E^{p, \mu}(F) \sse \bigcap_{n \geq 1} B_j^p(F_{\mu} ^n)$, which by (8), implies 
    \[ 
        E^{p, \mu}(F) \sse  B_j^p\left(\bigcap_{n \geq 1}F_{\mu}^n \right) = B_j^p\left(E^{p, \mu}(F)\right).
    \]
    
    This shows that $E^{p, \mu}(F)$ is an evident $(p, \ \mu)$-belief.
    
    \vspace{1 ex}
    
    It also follows that $E^{p, \mu}(F) \sse F_{\mu}^1 \sse B_{j}^p(F)$ for all $j \text{ in some } J \sse I$ with $|J| \geq \mu|I|$.
    
    \vspace{1 ex}
    
    \item If $\omega \in E^{p, \mu}(F)$, then by (1), $F$ is common $(p, \ \mu)$-belief at $\omega$.
    
    \vspace{1 ex}
    
    On the other hand, suppose that $F$ is common $(p, \ \mu)$-belief at $\omega$. 
    
    \vspace{1 ex}
    
    Let the event $E$ (with $\omega \in E$) and $J \sse I$ be as in the definition of common $(p, \ \mu)$-belief. 
    
    \vspace{1 ex}
    
    We will show, by induction, that $E \sse F_{\mu}^n$ for all $n \geq 1$. 
    
    \vspace{1 ex}
    
    For the base case, it follows from the definition of common $(p, \ \mu)$-belief that for each $j \in J$, $E \sse B_j^p(F)$, which implies that $E \sse F_{\mu}^1$.
    
    \vspace{1 ex}
    
    Now, suppose $E \sse F_{\mu}^n$. Then using the fact that $E$ is an evident $(p,\mu)$-belief, there must exist some $J' \sse I$ with with $|J| \geq \mu |I|$ such that for each $j \in J'$, $E \subseteq B_j^{p}(E)$.
    
    \vspace{1 ex}
    
    Applying (7) for each $j \in J'$ we have that $E \sse B_j^p(E) \sse B_j^p(F_{\mu}^n)$ for each $j \in J'$,
    
    which implies that 
    \[
        E \sse \bigcap_{j \in J'} B_{j}^p(F_{\mu}^n) \sse F_{\mu}^{n + 1}.
    \]
\end{enumerate}
\end{proof}

\section{Proof of Proposition~\ref{prop:hardness}}
\label{appendix:hardness}
\begin{proof}
We will reduce a classic NP-complete problem, \textsc{Clique}, to \textsc{Revolt}. 

\begin{definition}[{\textsc{Clique}}] \hfill

\textbf{Given:} $(G, k)$, where $G = (V, E)$ is an (undirected) graph with $|V| = n$ and $k \in \mathbb{N}_{ > 0}$.

\textbf{Question:} Does $G$ contain a clique of size $k$?
\end{definition}

Given an instance of \textsc{Clique} $(G, k)$, we define a prior $P$, summarized in the tables below:

\vspace{1 ex}

\begin{table}[h]
\begin{subtable}[t]{.48\textwidth}
   \centering
   \caption{Summary of the parameters in the prior $P$.}
   {\renewcommand{\arraystretch}{1.25}
\begin{tabular}{ |c|c| }
 \hline
 Parameter & Value\\
 \hline
  $p$ & $1$ \\
  \hline
  $\mu$ & $\frac{k}{n}$\\
  \hline
\end{tabular}
}
   \label{tab:prior_params}
\end{subtable}
\begin{subtable}[t]{.48\textwidth}
   \centering
   \caption{Summary of the distributions in the prior $P$.}
   {\renewcommand{\arraystretch}{1.25}
\begin{tabular}{ |c|c|c| } 
 \hline
  & $s = A$ & $s = B$ \\
  \hline
  $\Pr[\text{State} = s]$ & $\frac{1}{2}$ & $\frac{1}{2}$\\
  \hline
  $\Pr[\alpha | \text{State} = s]$ & 0 & 0\\
  \hline
  $\Pr[\chi | \text{State} = s]$ & $0.99$ & $0$\\
  \hline
  $\Pr[\nu | \text{State} = s]$ & $0.01$ & $1$\\
 \hline
\end{tabular}
}
   \label{tab:D_Ts}
\end{subtable}
\end{table}

Then, we consider an instance of \textsc{Revolt} with input $(G, P, \mu^* = \frac{k}{n}, q^* = \frac{0.99^k}{2})$. (Clearly, defining this instance can be done in polynomial time with respect to $n$.)

Now, suppose that $G$ contains a clique of size $k$. Then, with probability at least $q^*$, all of the agents that form the clique will be of type $\chi$. When this is the case, they will each observe the types of their neighbors, and will believe with probability 1 that all $k$ agents that form the clique (including themselves) feel secure enough to revolt. As a result, a revolt of size at least $\frac{k}{n} = \mu^*$ is supported with probability at least $q^*$ in $G$ under the prior $P$.

On the other hand, suppose that a revolt of size at least $\mu^*$ is supported with probability at least $q^*$ in $G$ under the prior $P$. That is, suppose the event that at least $k$ agents feel secure enough to revolt occurs (call this event $E$) with probability at least $q^* > 0$. 

For $E$ to occur, at least $k$ agents of type $\chi$ must believe with probability 1 that at least $k - 1$ other agents (of type $\chi$) feel secure enough to revolt. This is only possible, though, when the existence of at least $k$ agents of type $\chi$ is common knowledge among a group of at least $k$ $\chi$-type agents. 

As a result of the limited information available to each agent, common knowledge of the existence of $k$ agents of type $\chi$ can only occur within a clique of $k$ agents. This follows from that fact that in order to be certain of another agent's reasoning about some observed information (here, the type of some agent or agents), an agent $i$ must observe that the other agent $j$ observes the same information and also observe that $j$ observes that $i$ observes the information, and so on.

In order for $E$ to occur, this kind of mutual reasoning must happen among $k$ $\chi$-type agents, each of whom would need to be a neighbor of the other $k - 1$ to facilitate the mutual observation of each others' types. That is, they would need to form a clique of $k$ agents.\footnote{See Chwe's paper \cite{Chwe1999} for a similar setting where cliques are necessary for common knowledge.} $E$ occurs with non-zero probability, so there must exist a clique of size $k$ in $G$.
\end{proof}

We note here that in addition to establishing NP-hardness, this proof demonstrates that in order to achieve the exactness required to solve \textsc{Revolt}, an algorithm must use more information than just agent degrees, which, as we show, are sufficient for the relaxation \textsc{Promise Revolt}. 

We also conjecture that \textsc{Revolt} is $\#\text{P}$-hard.

\section{A Chernoff-Hoeffding Bound for Settings with Local Dependence}
\label{appendix:dependent_bound}
For a standard Chernoff-Hoeffding Bound, we consider taking the sum of independent random variables. However, similar bounds exist for sums of dependent random variables when the dependence between the variables is constrained. Borrowing from a book on concentration of measure by Dubhashi and Panconesi \cite{Dubhashi2009}, we consider the case where dependencies between the variables are encoded in a \textit{dependency graph} $\Gamma$.

Formally, let $\Gamma$ be a graph on $n$ vertices such that for each vertex $i \in [n]$, the following property holds: if $i$ is not adjacent to a distinct vertex $j \in n$ in $\Gamma$, then $X_i$ is independent from $X_j$.

Let $\chi^*(\Gamma)$ denote the \textit{fractional chromatic number} of $\Gamma$, which represents the smallest possible ratio $\frac{a}{b}$ of positive integers $a, b$ such that the vertices of $\Gamma$ can each be assigned a set of $b$ colors from a total set of $a$ colors under the constraint that adjacent vertices must be assigned disjoint sets of colors. (A formal definition of the fractional chromatic number can be found in Dubhashi and Panconesi \cite{Dubhashi2009}.)

\begin{theorem}[Theorem 3.2 in Dubhashi and Panconesi \cite{Dubhashi2009}]
\hfill

 Suppose $X = \sum_{i = 1}^n X_{i}$ where $0 \leq X_{i} \leq 1$ for each $i$. Then, for $t > 0$,
\[
    \Pr[X \geq \E[X] + t], \Pr[X \leq \E[X] - t] \leq \exp \left(\frac{-2t^2}{\chi^*(\Gamma) n}\right).
\]

\label{theorem:dependent_bound}
\end{theorem}

\end{document}